\documentclass[11pt,leqno,fleqn]{article}

\usepackage{graphicx}
\usepackage{enumitem}
\newcommand{\Comment}[1]{\relax}
\usepackage{comment}
\usepackage{caption}
\usepackage{subcaption}
\usepackage{amsthm}
\usepackage{latexsym}
\usepackage{amsmath}
\usepackage{amssymb}
\usepackage{amsfonts}
\usepackage{mathtools}
\usepackage{algorithm}
\usepackage{algorithmicx}
\usepackage{amsbsy}
\usepackage{synttree} 
\usepackage{diagrams}
\usepackage{enumitem} 
\usepackage[usenames]{color}
\usepackage[colorlinks=false, %
           citecolor=red,filecolor=blue,linkcolor=blue,urlcolor=blue,
           linktoc=page,
           pagebackref=true
           ]{hyperref} 

\usepackage{fullpage} 
\usepackage{afterpage}  
\usepackage{float} 
\usepackage{wrapfig}  
\usepackage{times}
\usepackage{bm}
\usepackage{stmaryrd} 
\usepackage{MnSymbol}  
\usepackage{mathrsfs}  
\usepackage[all]{xy}
\usepackage{fancybox}
\usepackage{alltt}

\newcommand{\Hide}[1]{}

\newif
\ifnote
\notetrue

\newif
\ifTR
\TRfalse

\newtheorem{theorem}{Theorem}
\numberwithin{theorem}{section}
\newtheorem{lemma}[theorem]{Lemma}
\newtheorem{corollary}[theorem]{Corollary}

\newtheorem{propxxx}[theorem]{Property}

\newtheorem{restxxx}[theorem]{Restriction}

\newtheorem{agreexxx}[theorem]{Agreement}

\newtheorem{termxxx}[theorem]{Terminology}

\newtheorem{notxxx}[theorem]{Notation}

\newtheorem{claimxxx}[theorem]{Claim}

\newtheorem{assumxxx}[theorem]{Assumption}

\newtheorem{exaxxx}[theorem]{Example}

\newtheorem{remxxx}[theorem]{Remark}

\newtheorem{openxxx}[theorem]{Open Problem}
\newtheorem{conjxxx}[theorem]{Conjecture}

\newtheorem{defxxx}[theorem]{Definition}
\newenvironment{definition}[1]{\begin{defxxx}[\emph{#1}]\rm}%
{\hfill\QED\end{defxxx}}
\newtheorem{procxxx}[theorem]{Procedure}
{\hfill\QED\end{procxxx}}

\newtheorem{Prxxx}[theorem]{Proof}
{\end{Prxxx}} 

  {\addtolength{\leftskip}{#1}\addtolength{\rightskip}{#2}}{\par}

\newcommand{\argmin}{\arg\!\min}
\newcommand{\argmax}{\arg\!\max}

\newcommand{\Set}[1]{\{ #1 \}}

\newcommand{\G}{{\cal G}}

\newcommand{\B}{{\cal B}}

\newcommand{\Let}[3]%
    {\textbf{\textsf{let}}\ {#1}\,{#2}\ \textbf{\textsf{in}}\;{#3}\,}
\newcommand{\Try}[3]%
    {\textbf{\textsf{try}}\ {#1} {#2}\ \textbf{\textsf{in}}\;{#3}\;}
\newcommand{\Mix}[3]%
    {\textbf{\textsf{mix}}\ {#1} {#2}\ \textbf{\textsf{in}}\;{#3}\;}
\newcommand{\LET}[3]%
    {\textbf{\textsf{let}}^*\ {#1} {#2}\ \textbf{\textsf{in}}\;{#3}\;}
\newcommand{\Letrec}[3]%
    {\textbf{\textsf{letrec}}\ {#1} {#2}\ \textbf{\textsf{in}}\;{#3}\;}

\newcommand{\QED}{{\Large $\square$}} 

\newcommand{\spacing}[2]{
  \renewcommand{\baselinestretch}{#2}
  \small\normalsize #1
  \setlength{\parskip}{0.1\baselineskip}
  \settowidth{\parindent}{xxxx}
  \setlength{\parindent}{#2\parindent}
  \setlength{\leftmargini}{\parindent}
  \setlength{\leftmarginii}{\parindent}
  \setlength{\leftmarginiii}{\parindent}
  \setlength{\footnotesep}{#2\footnotesep}
}

\begin{document}

\spacing{\normalsize}{0.98}
\setcounter{page}{1}     
\setcounter{tocdepth}{1} 
\ifTR
  \pagenumbering{roman} 
\else
\fi
\title{Average Delay of Routing Trees}

\author{ Saber Mirzaei \\
       Boston University  \\
        \ifTR Boston, Massachusetts \\
        \href{mailto:smirzaei@bu.edu}{smirzaei{@}bu.edu}
        \else \fi
}

   \date{\today}

\maketitle
  \ifTR
     \thispagestyle{empty} 
  \else
  \fi

\vspace{-.3in}
  \begin{abstract}
  The general communication tree embedding problem is the problem of mapping a set of communicating terminals,
represented by a graph $G$, into the set
of vertices of some physical network represented by a tree $T$.
In the case where
the vertices of G are mapped into the leaves of the host tree $T$ 
the underlying tree is called a \emph{routing tree}
and
if the internal vertices of $T$ are forced to have degree $3$, 
the host tree is known as \emph{layout tree}.
Different optimization problems have been studied in the class of
communication tree problems such as well-known minimum edge \emph{dilation}
and minimum edge \emph{congestion} problems.
In this report we study the less investigate measure \emph{i.e.}
\emph{tree length}, which is a representative for
average edge dilation (communication delay) measure and also for average
edge congestion measure. 
We show that finding a routing tree $T$ for an arbitrary graph $G$ with
minimum tree length is an NP-Hard problem.
  \end{abstract}

\ifTR
    \newpage
    \tableofcontents
    \newpage
    \pagenumbering{arabic}
\else
    \vspace{-.2in}
\fi

\section{Definitions and Introductory Points}
\label{sect:def-not}
  Consider a group of terminals communicating via a finite network $G=(V, E)$,
where the set of vertices (finite set $V$) and edges (finite set $E$), respectively represent the
collection of terminals and their direct communication paths.
We show $|V|$ by $n$ and $|E|$ as $m$.

\noindent
The general communication tree embedding problem is the problem of mapping the set of terminals
into the set of vertices of some physical network represented by a tree $T$.
Accordingly, the two vertices $v, u \in V(G)$ that are directly connected via $e \in E(G)$,
are connected indirectly via some path $P_T(v,u)$ in $T$.
In the case where the vertices of $G$ are mapped to the leaves of the host tree,
the underlying tree is called a \emph{routing tree}.
In this report we mostly focus on the case where the internal vertices of the
host tree have degree $3$ (known as \emph{tree layout} problem).
We denote the sets of leaf nodes and internal nodes of tree $T$ respectively by $V_L(T)$ and $V_I(T)$.
\footnote{We try to use the term node in case of trees as opposed to the term vertex, which we use for general graphs.}

For a graph $G$ and a communication tree $T$ for $G$, there are different measures defined in literature.
In following we define the two measures that we are intrusted in this report.
For a comprehensive list of measures, an interested reader can refer to~\cite{i2001layout}.
\begin{definition}{Edge Dilation}
Consider a graph $G$ and a communication tree $T$ and a bijection $\varphi$ from vertices of $G$ to leaf nodes of $T$. The \emph{dilation} $\lambda(uv,T,\varphi,G)$ of an edge $\Set{u,v} \in E(G)$
is the distance between $\varphi(u)$ and $\varphi(v)$ in $T$.
\end{definition}
\noindent
We represent the distance of two vertices $\Set{u,v}$ in a graph $G$ with $d_G(u,v)$
\begin{definition}{Edge Congestion}
Give a graph $G$ and a communication tree $T$ and and a bijective mapping $\varphi: V(G) \rightarrow V_L(T)$ . The \emph{congestion} $\delta(xv,T,\varphi,G)$ and of an edge $\Set{x,y} \in E(T)$
is the the number of edges in $\Set{u,v} \in E(G)$ that in $T$, the path $P_T(\varphi(v),\varphi(u))$ traverse trough $\Set{x,y}$.
\end{definition}
Based on the definition of the communication tree for a graph $G$, removal of every edge $\Set{x,y} \in E(T)$
partitions the set of vertices of $G$ into two component. Hence every edge of tree corresponds to a cut in
$G$. Therefore the congestion of $\Set{x,y} \in E(T)$ is the size of the cut it corresponds to.

\noindent
Several optimization problems can be defined based on these two measures.
\emph{Minimum tree layout dilation} is the problem of finding a tree layout
for a given graph $G$ such that the maximum edge dilation is minimized,
where the maximum is taken over all edge of $G$.
In~\cite{monien1985complexity} it is shown that the problem of finding a tree layout with minimum dilation is NP-hard, when the layout tree is rooted.

\noindent
Similarly, given a graph $G$, in \emph{minimum tree layout congestion} problem the goal is to find
a tree layout $T$, such that the maximum edge congestion is minimized.
In~\cite{seymour1994call} Seymour and Thomas show that the minimum tree layout congestion problem is polynomially solvable
for the case of planer graphs, and is NP-hard when considering general graphs.
In this report we study the \emph{minimum tree layout length} problem (shortly called Min Tree Length), formally defined as it follows.
\begin{definition}{Minimum tree layout length}
Consider the finite undirected graph $G=(V,E)$. The minimum tree layout length problem is the problem of
finding layout tree $T$ and a bijective mapping $\varphi: V(G) \rightarrow V_L(T)$ such that $L(T,\varphi,G) = \sum_{\Set{u,v} \in E(G)} \lambda(uv,T,\varphi,G)$ is minimized.
\end{definition}

\noindent
It is not hard to see that $\sum_{\Set{u,v} \in E(G)} \lambda(uv,T,\varphi,G) = \sum_{\Set{x,y} \in E(T)} \delta(xv,T,\varphi,G)$. Hence, in the rest of
this report we may use them interchangeably.

Accordingly, in the communication graph embedding problems,
the dilation of an edge $\Set{u,v} \in E(G)$
abstractly represent the communication delay between vertices $u$ and $v$.
Similarly the congestion of an edge $e \in E(T)$ is a representative for the
traffic on the physical link $e$. Hence tree length measure corresponds to the
average delay between the vertices of $G$ and also to the average edge congestion of the host tree.

\section{Minimum Length of Tree Layout}
\label{sect:Tree-Lenght}
  In the special case of tree layout problem, the underlying host graph is a tree $T$ where the
degree of every node is either 1 or 3 and the vertices of $G$ are being mapped to
leaves of $T$.
In this section we study minimum tree layout length. We show that Min Tree Length problem
is NP-hard for multi-graphs\footnote{By multi-graph we refer to finite graphs with possibility of parallel edges and no loop.},
and later on we show the problem stays NP-hard when restricted to the class of simple graphs.
\subsection{Min Tree Length of Complete Graphs}
\label{sec:complete-graph-optimal-layout}
Consider the complete graph $G=(V,E)$ where $\forall u,v \in V(G), \Set{u,v} \in E(G)$.
It is not hard to see that a layout tree $T$ is a solution for the Min Tree Layout problem for $G$,
iff $|V_L(T)|=n$ (and hence $|V(T)| = 2n-1$), and the summation of distance of leaf nodes of $T$ is minimized.
We denote the summation of distances of leaves of a tree $T$ by:
\begin{align*}
    \sigma_{LL}(T) = \dfrac{1}{2} \sum_{x,y \in V_L(T)} d_T(x,y)
\end{align*}

\noindent
Leaf to leaf distance summation measure $\sigma_{LL}$ is very similar to the definition of \emph{Wiener index} (proposed by chemist
Wiener~\cite{wiener1947structural}),
which is the summation of distances of all vertices of a given graph $G$ as represented in following equation.
\begin{align*}
    \sigma(G) = \dfrac{1}{2} \sum_{u,v \in V(G)} d_G(u,v)
\end{align*}
\noindent
Wiener index is widely studied both in mathematical and chemical literature.
In~\cite{fischermann2002wiener} Fischermann \emph{et al.} study Wiener index of trees.
In this works authors represent the structure of the family of the trees
that have minimum (or maximum) Wiener index among all the trees of the same order with
maximum node degree $\Delta \ge 3$.
Due to similarity of $\sigma_{LL}$ measure and Wiener index, in the rest of this subsection we borrow some of the notations and definitions
from~\cite{fischermann2002wiener} in order to study $\sigma_{LL}$ for trees with maximum degree $\Delta$.

\begin{definition}{$\mathfrak{T}(\mathcal{R}, \Delta)$ tree family}
\label{def:opt-family-structure}
Consider integers $\Delta$ and $\mathcal{R} \in \Set{\Delta, \Delta - 1}$.
For a given $n \in \mathbb{N}$, the family $\mathfrak{T}(\mathcal{R}, \Delta)$ of trees
with $n$ nodes has a unique member $\mathfrak{T}$ up to isomorphism, defined using
a planar embedding as it follows.

\noindent
Let $M_k(\mathcal{R}, \Delta) \le n < M_{k+1}(\mathcal{R}, \Delta)$ where:
\begin{align*}
  M_k(\mathcal{R}, \Delta) = \begin{cases}
      1 & \quad\text{$k=0$} \\
      1+  \mathcal{R} + \mathcal{R} \times (\Delta - 1) + \ldots + \mathcal{R} \times (\Delta - 1)^{k-1}& \quad\text{$k \ge 1$}
    \end{cases}
\end{align*}
\noindent
Figure~\ref{fig:opt-tree-embedding} depicts the embedding of tree $\mathfrak{T}$ with the following properties:
\begin{enumerate}
  \item all nodes of $\mathfrak{T}$ lie on some line $\textbf{L}_i$ for $0 \le i \le k+1$
  \item exactly one node $v$ lies on the line $\textbf{L}_0$ which has $\min \Set{n-1,\mathcal{R}}$ children on line $\textbf{L}_1$
  \item for $i = 1$ to $k-1$ every node on line $\textbf{L}_i$ is connect to $\Delta - 1$ nodes on line $\textbf{L}_{i+1}$
  and one node on line $\textbf{L}_{i-1}$
  \item the only line that may be incomplete, is line $\textbf{L}_{k+1}$. Let $n - M_k(\mathcal{R}, \Delta) = m\times (\Delta - 1) + r$
  for $0 \le r < \Delta -1$, where $n - M_k(\mathcal{R}, \Delta)$ is the number of remaining nodes on line $\textbf{L}_{k+1}$.
  Also let $\Set{v_1,..v_{m+1}}$ be the set of $m$ left most nodes on line $\textbf{L}_{k}$ where $v_{m+1}$ is the right most one in the set.
  Each of $v_1,..v_m$ nodes is connected to $\Delta - 1$ nodes on line $\textbf{L}_{k+1}$, while $v_{m+1}$ is connected to $r$ nodes from
  line $\textbf{L}_{k+1}$ (see figure~\ref{fig:opt-tree-embedding}).
\end{enumerate}
\end{definition}
\begin{figure}[h]
    \begin{center}
        \includegraphics[scale=0.9]{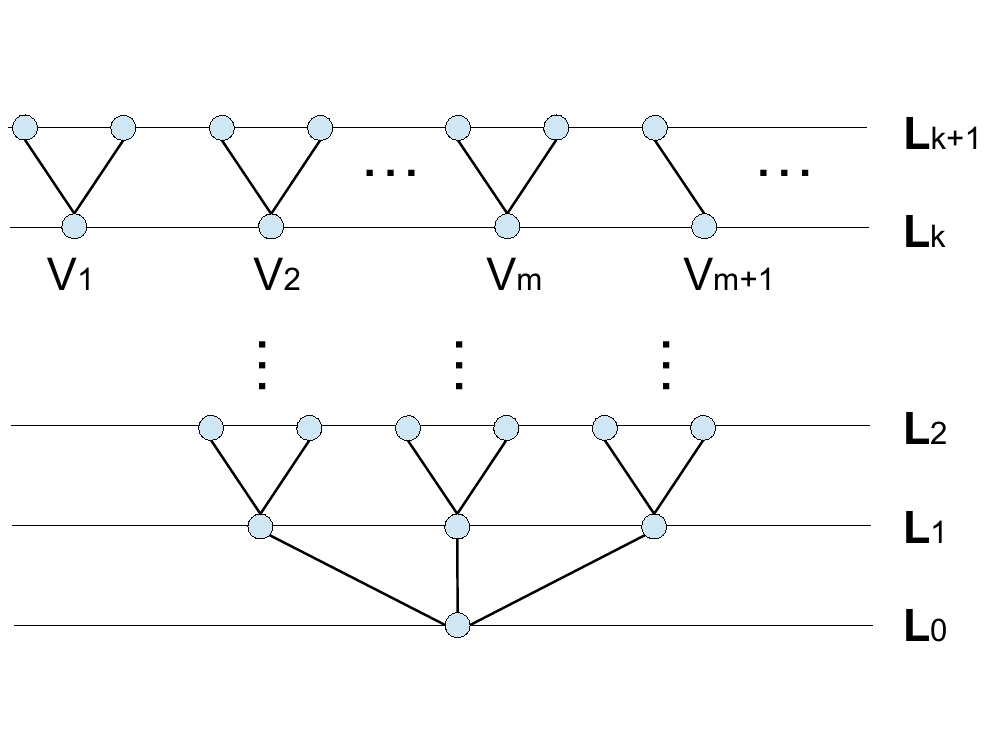}
    \end{center}
    \vspace{-0.2in}
    \caption{Planar embedding of $\mathfrak{T}$ represented as the embedding of nodes on $k+2$ lines $\textbf{L}_i$ for $0 \le i \le k+1$.}
    \label{fig:opt-tree-embedding}
\end{figure}
We defined the family $\mathfrak{T}(\mathcal{R}, \Delta)$ for the general case of trees with maximum degree $\Delta$,
while we focus on the case where the degree of every internal node is $\Delta = 3$, but all the results presented
in the rest of this subsection extend to all trees with arbitrary max degree $\Delta \ge 3$.
\begin{lemma}
\label{lemma:same-height-opt-if-in-family}
Consider tree $T \in \mathfrak{T}(\mathcal{R}, \Delta)$ of order $n$, and assume
$M_k(\mathcal{R}, \Delta) < n < M_{k+1}(\mathcal{R}, \Delta)$ for $k \ge 1$.
Let $\overline{T}$ be an arbitrary tree of order $n$ constructed from tree $\overline{T}_0 \in \mathfrak{T}(\mathcal{R}, \Delta)$,
with $M_k(\mathcal{R}, \Delta)$ nodes, by attaching $n - M_k(\mathcal{R}, \Delta)$ nodes to the leaf nodes of $\overline{T}_0$
(which lie on the line $\textbf{L}_k$). Then it is the case that either $\sigma_{LL}(\overline{T}) > \sigma_{LL}(T)$
or $\overline{T}$ is isomorphic with $T$.
\end{lemma}
\begin{proof}
We proof this lemma using induction on the height of tree $\overline{T}$, when embedded on the plane as explained in definition~\ref{def:opt-family-structure}.
It is easy to check the correctness of theorem for trees of height 1 and 2.
Assume tree $\overline{T}$ of height $k$ and tree $T \in \mathfrak{T}(\mathcal{R}, \Delta)$
are not isomorphic. Let $v \in V(\overline{T})$ be the node on line $\textbf{L}_0$.
Node $v$ is connected to $\mathcal{R}$ subtrees $\Set{\overline{T}_1,\ldots, \overline{T}_{\mathcal{R}}}$.
Based on the assumption of induction every subtrees $\overline{T}_i$ is a member
of the family $\mathfrak{T}(\Delta-1, \Delta)$ of height $k$ or $k-1$.
Since $\overline{T}$ and $T$ are not isomorph, there are at least two subtrees $\overline{T}_i$ and $\overline{T}_j$ for $i \neq j$
where are incomplete on line $\textbf{L}_{k}$.
Formally speaking $\exists 1 \le i \neq j \le \mathcal{R}, \overline{T}_i, \overline{T}_j \in \mathfrak{T}(\Delta-1, \Delta)$ and
$|V(\overline{T}_i)| - M_{k-1}(\Delta-1,\Delta) = r_i > 0 \wedge |V(\overline{T}_j)| - M_{k-1}(\Delta-1,\Delta) = r_j > 0$.
Without loss of generality we assume $i = 1$ and $j = 1$ and also $|V(\overline{T}_1)| \le |V(\overline{T}_2)|$.
Figure~\ref{fig:not-isomorph-with-optimal} abstractly represents tree $\overline{T}$.

\begin{figure}[h]
        \centering
        \begin{subfigure}[b]{0.3\textwidth}
               \includegraphics[scale=.65]
                {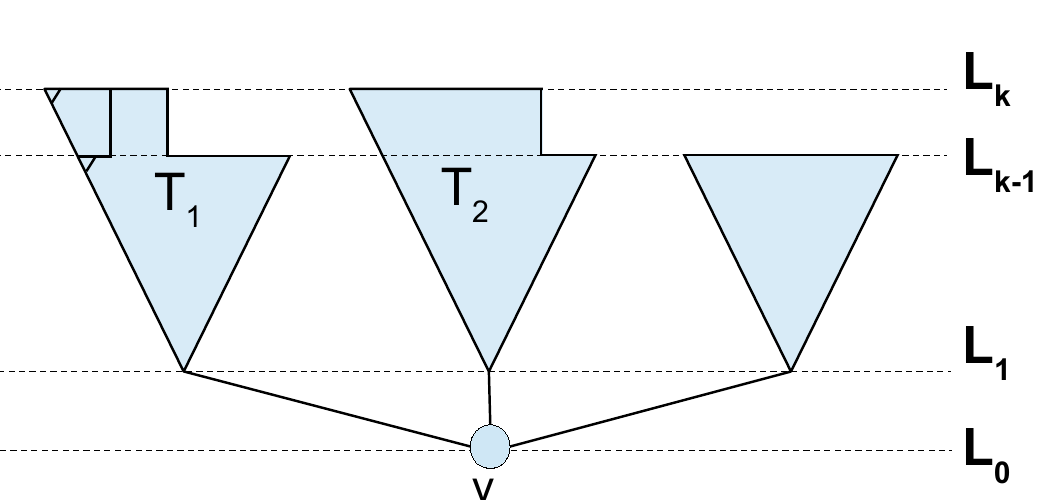}
                \caption{Initial tree $\overline{T}$.}
                 \label{fig:not-isomorph-with-optimal}
        \end{subfigure}
        \qquad\qquad\qquad\qquad
        \begin{subfigure}[b]{0.3\textwidth}
                \includegraphics[scale=.65]
                {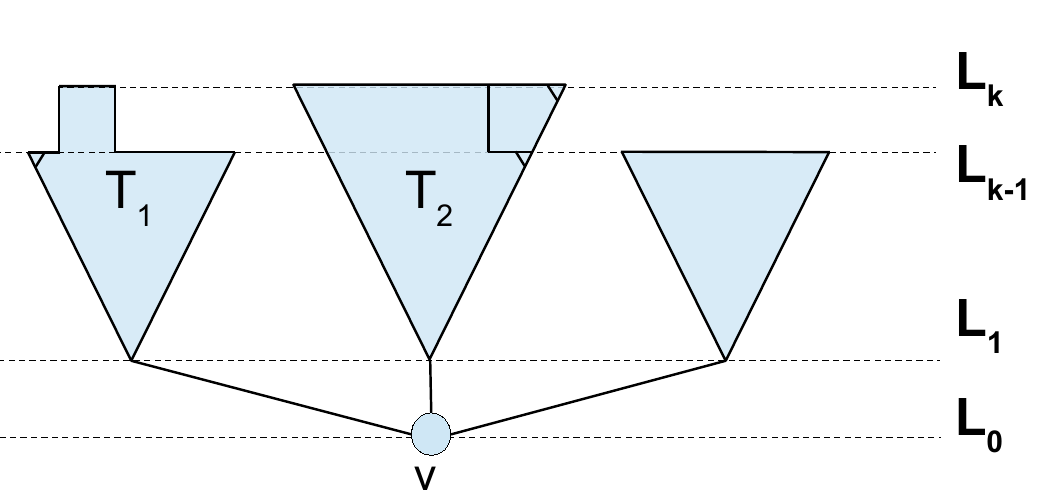}
                \caption{Tree $\overline{T}^{\prime}$ constructed from $\overline{T}$.}
                \label{fig:not-isomorph-with-optimal-after}
        \end{subfigure}
        \caption{
        Figure~\ref{fig:not-isomorph-with-optimal} represents tree $\overline{T} \notin \mathfrak{T}(\mathcal{R}, \Delta)$ of height $k+1$, constructed from tree $\overline{T}_0 \in \mathfrak{T}(\mathcal{R}, \Delta)$ of height $k$. Node $v$ is connected to at least two subtrees $\overline{T}_i,\overline{T}_j \in \mathfrak{T}(\mathcal{R}, \Delta)$ of height
        $k$ which are incomplete on line $\textbf{L}_{k}$.
        Alternative tree $\overline{T}^{\prime}$ depicted in figure~\ref{fig:not-isomorph-with-optimal-after}, is constructed by relocating some leaf nodes of subtree $\overline{T}_1$ that are on line
        $\textbf{L}_{k}$. As you see in this construction $r_2$ left most leaves of $\overline{T}_1$ are relocated to complete subtree $\overline{T}_2$. In this example the number of leaf nodes of $\overline{T}_1$ on $\textbf{L}_{k}$, is large enough to complete the subtree $\overline{T}_2$.
        }
    \label{fig:not-isomorph-with-optimal-before-after}
\end{figure}

\noindent
Let $\mathcal{\textbf{L}}_l(T)$ denote the set of nodes of tree $T$ on line $l$. We present an alternative tree $\overline{T}^{\prime}$, by relocating some leaf nodes of $\mathcal{\textbf{L}}_k(\overline{T}_{1})$ (in order from left to right) to complete the line $\textbf{L}_{k}$ of $\overline{T}_2$ (in order from right to left).
Let $\mathcal{L} \subseteq \mathcal{\textbf{L}}_k(T_{1})$ be the
set of leaf nodes of $\overline{T}_{1}$ on line $\textbf{L}_{k}$, candidate for relocation. Figure~\ref{fig:not-isomorph-with-optimal-after} depicts tree $\overline{T}^{\prime}$ constructed from $T$.

\noindent
In the alternative tree $\overline{T}^{\prime}$, consider a bijective mapping $\varphi_0$ from the nodes of $\overline{T}_{1}$ to nodes of subtree $\overline{T}^{\prime}_{2}$, where $\overline{T}^{\prime}_{2}$ is the modified version of $\overline{T}_{2}$ in $\overline{T}$. More specifically, mapping $\varphi_0$ is a reflection form $\overline{T}_{1}$ to $\overline{T}_{2}$, which reflects nodes of $\mathcal{\textbf{L}}_l(\overline{T}_{1})$ to the nodes of $\mathcal{\textbf{L}}_{l-1}(\overline{T}^{\prime}_2)$ for $2 \le l \le k$, such that the left most node on line $\textbf{L}_{l}$ of $\overline{T}_{1}$ is mapped to the right most node of $\overline{T}^{\prime}_2$ on line $\textbf{L}_{l}$.
Accordingly $\varphi_0$ maps every leaf node of $\overline{\mathcal{L}} = \mathcal{\textbf{L}}_k(\overline{T}_{1}) - \mathcal{L}$ (the set of remaining leaf nodes of $\overline{T}_{1}$ on line $\textbf{L}_k$) to one leaf node of $\overline{T}_2$ on line $\textbf{L}_{k}$.
On the other hand every leaf node of $\mathcal{\textbf{L}}_{k-1}(\overline{T}_{1})$ is mapped to one node (leaf or internal) of $\overline{T}_2$ on line $\textbf{L}_{k-1}$.

\noindent
From $\varphi_0$ we construct a bijective mapping $\varphi$ from leaf nodes of $\overline{T_1}$ to sets of leaf nodes of $\overline{T}^{\prime}_2$.
For every $w \in V_L(T)$, $\varphi(w) = \Set{\varphi_0(w)}$ if $\varphi_0(w) \in V_L(\overline{T}^{\prime}_2)$ is a leaf node node, otherwise ($\varphi_0(w)$ is an internal node of $\overline{T}^{\prime}_2$ on line $\textbf{L}_{k-1}$), $\varphi$ maps $w$ to the set of direct children of $\varphi_0(w)$ on line $\textbf{L}_{k}$ \footnote{In this case $|\varphi(w)| \le \Delta - 1$.}.

\noindent
Using the bijective mapping $\varphi$, we analyze the change in value of $\sigma_{LL}$ in the process of constructing $\overline{T}^{\prime}$ from
$\overline{T}$ as it follows.
\begin{enumerate}
  \item Clearly the internal summation of distances of nodes in $\mathcal{L}$ stays unchanged.
  \item For every leaf node $w_1 \in \mathcal{L}, w_2 \in V_L(T) - V_L(T_1)\uplus V_L(T_2)$, it is the case that $d_{\overline{T}}(w_1,w_2) = d_{\overline{T}^{\prime}}$. Hence, the summation of distances among nodes in $\mathcal{L}$ and leaf nodes of $V_L(\overline{T}) - V_L(\overline{T}_1)\uplus V_L(\overline{T}_2)$ also does not change.
  \item We show that for every $w_1 \in \mathcal{L}$ the summation of distances of $w_1$ from leaf nodes of $V_L(\overline{T}_2)\uplus V_L(\overline{T}_1) - \Set{w_1}$ is
  greater than the summation of distances of $\varphi(w_1)$ from leaf nodes of $V_L(\overline{T}^{\prime}_1)\uplus V_L(\overline{T}^{\prime}_2) - \Set{\varphi{w_1}}$.
  \noindent
  \begin{notxxx}
        Let $\mathcal{L}_1, \mathcal{L}_2 \subset V(T)$, for an arbitrary tree $T$. By $\sigma_T(\mathcal{L}_1,\mathcal{L}_2)$ we denote the summation of distances of nodes of $\mathcal{L}_1$ and $\mathcal{L}_1$. Formally speaking:
        \begin{align*}
        \sigma_T(\mathcal{L}_1,\mathcal{L}_2) = \dfrac{1}{2} \sum_{v \in \mathcal{L}_1} \sum_{u \in \mathcal{L}_2} d_T(v,u)
        \end{align*}
  \end{notxxx}

  \noindent
  For every $w_2 \in V_L(\overline{T}_1) - \mathcal{L}$:
  \begin{itemize}
    \item If $w_2$ is on line $\textbf{L}_k$, it is the case that $d_{\overline{T}}(w_1,w_2) = d_{\overline{T}^{\prime}}(\varphi(w_1),\varphi(w_2))$. Therefore:
        \begin{align}
        \label{eq:lemma-internal}
        \sigma_{\overline{T}}(\Set{w_1},\Set{w_2}\uplus \varphi(w_2)) = \sigma_{\overline{T}}(\Set{\varphi(w_1)},\Set{w_2} \uplus \varphi(w_2))
        \end{align}
    \item If $w_2$ is on line $\textbf{L}_{k-1}$ and $\varphi(w_2)$ maps $w_2$ to a leaf node of $\overline{T}^{\prime}_2$ on line $\textbf{L}_{k-1}$,
    similar to previous case we have:
    \begin{align}
    \label{eq:lemma-others}
        \sigma_{\overline{T}}(\Set{w_1},\Set{w_2}\uplus \varphi(w_2)) = \sigma_{\overline{T}}(\Set{\varphi(w_1)},\Set{w_2} \uplus \varphi(w_2))
    \end{align}
    \item Otherwise $w_2$ is on line $\textbf{L}_{k-1}$ and $\varphi(w_2)$ maps $w_2$ to a non-empty set of leaf node of $\overline{T}^{\prime}_2$ on line $\textbf{L}_{k}$. Assume the size of this set is $1 \le \nabla \le \Delta$. Since for some nodes $w_2$ where $|\varphi(w_2)| = \nabla = \Delta - 1 > 1$, then for such $w_2$, the following equally holds:
    \begin{align}
    \label{eq:lemma-main}
        \sigma_{\overline{T}}(\Set{w_1},\Set{w_2}\uplus \varphi(w_2)) - \sigma_{\overline{T}^{\prime}}(\Set{\varphi(w_1)},\Set{w_2} \uplus \varphi(w_2)) =  &\\
        \Big( d_{\overline{T}}(w_1,w_2) + 2k \times \nabla \Big) - \Big( 2k + (d_{\overline{T}^{\prime}}(\varphi(w_1),\varphi_0(w_2))+1) \times \nabla \Big) =  &\nonumber\\
        \Big( d_{\overline{T}}(w_1,w_2) + 2k \times \nabla \Big) - \Big( 2k + (d_{\overline{T}}(w_1,w_2)+1) \times \nabla \Big) =  &\nonumber \\
        (2k - d_{\overline{T}}(w_1,w_2)) \times (\nabla - 1) - d_{\overline{T}}(w_1,w_2)) \ge &\nonumber \\
        2k - 2d_{\overline{T}}(w_1,w_2)) > 0 &\nonumber
    \end{align}
  \end{itemize}
\end{enumerate}

\noindent
Putting the results of previous cases and equations~\ref{eq:lemma-internal},~\ref{eq:lemma-others} and ~\ref{eq:lemma-main}, we conclude that
$\sigma_{LL}(\overline{T}) > \sigma_{LL}(\overline{T}^{\prime})$.
If $\overline{T}^{\prime}_1 \notin \mathfrak{T}(\mathcal{R}, \Delta)$, then based on the assumption of induction, replacing the subtree $\overline{T}^{\prime}_1$ with subtree $ \overline{T}^{''}_1 \in \mathfrak{T}(\mathcal{R}, \Delta)$ of the same order, results in tree $\overline{T}^{''}$, where $\sigma_{LL}(\overline{T}^{''}) < \sigma_{LL}(\overline{T}^{\prime})$.

\noindent
Using the same approach and continuing with $\overline{T}^{''}$, in a sequence of leaf node relocations, we can construct the final tree $\widetilde{T}$, such that
in $\widetilde{T}$, node $v$ on line $\textbf{L}_0$ has exactly one incomplete subtree $\widetilde{T}_i$ where
$\widetilde{T}_i \in \mathfrak{T}(\Delta - 1, \Delta)$.
More specifically, $\forall 1 \le l < i$, all leaf nodes of $\widetilde{T}_l$ lei on line $\textbf{L}_{k}$ and
$\forall i < l \le \Delta$, all leaf nodes of $\widetilde{T}_l$ lei on line $\textbf{L}_{k-1}$, \emph{i.e.} $\widetilde{T} \in \mathfrak{T}(\mathcal{R}, \Delta)$.
\end{proof}
\begin{notxxx}
Given an arbitrary tree $T$, we define the \emph{planar line embedding} of $T$ similar to the the approach in the definition~\ref{def:opt-family-structure}.
Starting from a designated $v \in V(T)$, we embed $T$ on lines of plane, where $v$ lies on line $\textbf{L}_0$ and direct neighbours of $v$ are placed on line $\textbf{L}_1$.
Similarly all the nodes in distance $d$ from $v$ are placed on line $\textbf{L}_d$. Also for $u \in V(T)$ on line $\textbf{L}_l$, the subtree rooted
at $u$, where all its nodes are on lines $\textbf{L}_{l^{\prime}}$ for $l^{\prime} \ge l$ is denoted by $T_u$.
Formally speaking $w \in V(T_u)$ iff $u$ is on the shortest path from $w$ to $v$ on line $\textbf{L}_0$.
\end{notxxx}
\begin{notxxx}
Consider a tree $T$ of order $n$ and a node $v \in V(T)$ of degree $\nabla$. Removing $v$  form $T$ partitions $T$ into a set of subtrees $\Set{\mathcal{T}_1,\ldots,\mathcal{T}_{\nabla}}$. We call node $v$ \emph{central} if $\forall 1 \le i \le \nabla, |V(\mathcal{T}_i)| \le \dfrac{n}{2}$. The
set of central nodes of tree $T$ is represented by $\mathcal{C}_T$.
\end{notxxx}
\begin{theorem}
\label{thr:central-node}
Every arbitrary tree $T$ has at least one and at most two central nodes. In other word $1 \le |\mathcal{C}_T| \le 2$.
\end{theorem}
\noindent
For proof see~\cite{shiloach1979}. Using this theorem, we proof the main result of this subsection as we present
in the theorem~\ref{thr:opt-LL-tree}.
\begin{theorem}
\label{thr:opt-LL-tree}
Consider tree $T$ with maximum node degree $\Delta$.
where $\sigma_{LL}(T) \le \sigma_{LL}(T^{\prime})$ for every tree $T^{\prime}$ (that $|V_L(T^\prime)| = |V_L(T)|$). Then in the planar line embedding of $T$ with central node $v \in \mathcal{C}_T$ on fixed line $\textbf{L}_0$, it is the case that:
\begin{itemize}
  \item $T \in \mathfrak{T}(\Delta, \Delta)$
  \item $T_u \in \mathfrak{T}(\Delta-1, \Delta)$ for $u \neq v$
\end{itemize}
\end{theorem}
\begin{proof}

\noindent
The proof of this theorem is carried out using induction on the height of planar line embedding of tree $T$.
It is not hard to check the correctness of theorem for trees of height 1 and 2.
Let $T$ be a graph of height $h \ge 3$, and let $u_1,\ldots,u_{\Delta}$ be the direct neighbours of $v$ on line $\textbf{L}_0$ and $T_1,\ldots,T_{\Delta}$
respectively be their corresponding subtrees.

\noindent
\paragraph{Case 1: $\exists w_1, w_2 \in V_L(T), d_T(w_1,v) \ge d_T(w_2,v) + 2$.}
Based on the assumption of induction $w_1$ and $w_2$ can not be on the same subtree. Without loss of the generality
assume $w_1$ and $w_2$ are the two leaf nodes with maximum distance and $w_1 \in T_1$ and $w_2 \in T_2$. Let $T_{11},T_{12},\ldots, T_{1,\Delta -1} \subset T_1$ be subtrees respectively with roots $u_{11}, \ldots u_{1,\Delta -1}$ connected to $u_1$ (nodes $u_11, \ldots u_{1,\Delta -1}$ lie on line $\textbf{L}_2$).
Also assume $w_1 \in V_L(T_{11})$.

\noindent
\paragraph{Case 1.1: $|V_L(T_{11})| > |V_L(T_2)|$.}
We construct an alternative tree $\widetilde{T}$ by removing edges $\Set{u_1,u_{11}}$ and $\Set{v,u_2}$ and introducing
two new edges $\Set{v,u_{11}}$ and $\Set{u_1,u_2}$. The structures of initial tree $T$ and the alternative tree $\widetilde{T}$ are represented in figure~\ref{fig:unbalance-tree-before-after}.
\begin{figure}[h]
        \centering
        \begin{subfigure}[b]{0.3\textwidth}
                \includegraphics[scale=.7]
                {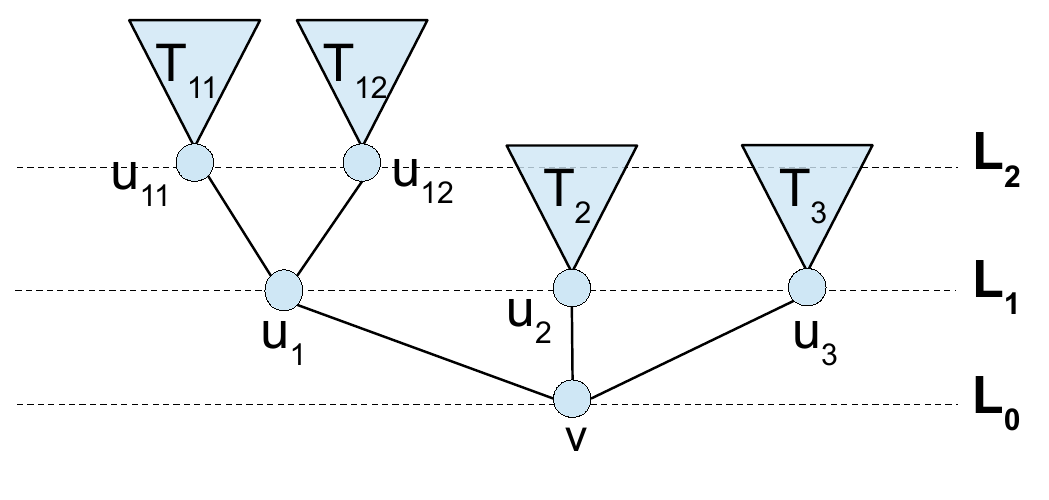}
                \caption{Initial tree $T$ where $d_T(w_1,v) \ge d_T(w_2,v) + 2$ for $w_1 \in V_L(T_{11})$ and $w_2 \in V_L(T_{2})$.}
        \end{subfigure}
        \qquad\qquad\qquad
        \begin{subfigure}[b]{0.3\textwidth}
                \includegraphics[scale=.7]
                {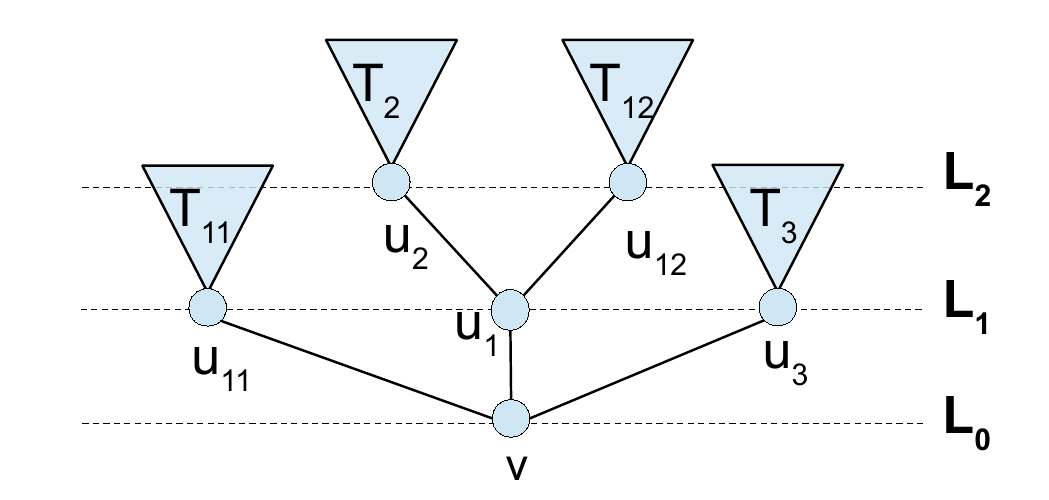}
                \caption{Alternative tree $\widetilde{T}$ constructed from $T$.}
        \end{subfigure}
        \caption{Representation of initial tree $T$ and its  modified version $\widetilde{T}$.}
    \label{fig:unbalance-tree-before-after}
\end{figure}
One can verify that the following equation~\ref{eq:unbalance-tree-before-after} correctly represents the relation between $\sigma_{LL}(T)$ and $\sigma_{LL}(\widetilde{T})$.
\begin{align}
\label{eq:unbalance-tree-before-after}
\sigma_{LL}(T) - \sigma_{LL}(\widetilde{T}) = \\
& |V_L(T_{11})|\times ( \sum_{2 \le i \le \Delta - 1}|V_L(T_{1i})| - \sum_{3 \le i \le \Delta}|V_L(T_{i})|) \nonumber \\
& +  |V_L(T_{2})|\times ( \sum_{3 \le i \le \Delta}|V_L(T_{i})| - \sum_{2 \le i \le \Delta - 1}|V_L(T_{1i})|) \nonumber \\
& = (|V_L(T_{11})| - |V_L(T_{2})|) \times ( \sum_{2 \le i \le \Delta - 1}|V_L(T_{1i})| - \sum_{3 \le i \le \Delta}|V_L(T_{i})|) \nonumber
\end{align}
\noindent
It is the case that $\sum_{2 \le i \le \Delta - 1}|V_L(T_{1i})| < \sum_{3 \le i \le \Delta}|V_L(T_{i})|$ otherwise it must be the case that $|V(T_1)| > \sum_{2 \le i \le \Delta}|V(T_{i})| > \dfrac{|V(T)|}{2}$, which is in contradiction with the centrality of node $v$.
Therefore $\sigma_{LL}(T) - \sigma_{LL}(\widetilde{T}) > 0$, which contradicts the optimality of $T$.

\noindent
\paragraph{Case 1.2: $|V_L(T_{11})| \le |V_L(T_2)|$.} Hence, $d_T(w_1,v) = d_T(w_2,v) + 2$ and $w_1 \in V_L(T_{11})$ is located on line $\textbf{L}_k$ and $w_2 \in V_L(T_2)$ lies on line $\textbf{L}_{k-2}$. Also non of subtrees $T_{11}$ and $T_2$ can be complete respectively on lines $\textbf{L}_k$ and $\textbf{L}_{k-1}$.

\noindent
Let $\mathcal{\textbf{L}}_l(T)$ denote the set of nodes of tree $T$ on line $l$. Similar to the proof of lemma~\ref{lemma:same-height-opt-if-in-family}, we present an alternative tree $\widetilde{T}$, by relocating some leaf nodes of $\mathcal{\textbf{L}}_k(T_{11})$ (in order from left to right) to complete the line $\textbf{L}_{k-1}$ of $T_2$ (in order from right to left).
Let $\mathcal{L} \subseteq \mathcal{\textbf{L}}_k(T_{11})$ be the
set of leaf nodes of $T_{11}$ on line $\textbf{L}_{k}$, candidate for relocation. Based on an exact reasoning as in lemma~\ref{lemma:same-height-opt-if-in-family} (case 3), which we omit, it can be inferred that the summation of distance of leaf nodes in
$\mathcal{L}$ from leaf nodes of $T_{11} \uplus T_{2}$ reduces going from $T$ to $\widetilde{T}$. Formally it can be deduced that:
\begin{align}
\label{eq:internal-change-case12}
\sigma_T(\mathcal{L}, V_L(T_{11})) + \sigma_T(\mathcal{L}, V_L(T_{2})) >
\sigma_{\widetilde{T}} (\widetilde{\mathcal{L}}, V_L(\widetilde{T}_{11}))+ \sigma_{\widetilde{T}} (\widetilde{\mathcal{L}}, V_L(\widetilde{T}_{2}))
\end{align}
\noindent
Where $\widetilde{T}_{11}$ and $\widetilde{T}_{2}$ respectively correspond to $T_{11}$ and $T_2$ after relocating leaf nodes of $\mathcal{L}$ (represented by $\mathcal{\widetilde{L}}$ in $\widetilde{T}$).

\noindent
On the other hand, relocating $\mathcal{L}$, increases the distance of every leaf node in $\mathcal{L}$
from leaf node of $T_{12},\ldots,T_{1,\Delta-1}$ by 1 unit, while it decreases the distance of every node of $\mathcal{L}$
from every leaf node of $T_3,\ldots,T_{\Delta}$. Since we assumed that $w_1 \in V_L(T_{11})$ and $w_2 \in V_L(T_2)$ have the maximum distance among
all leaf nodes, then:
\begin{align}
\label{eq:others-change-case12}
|V_L(T_{12})|+\ldots+|V_L(T_{1,\Delta-1})| < |V_L(T_{3})|+\ldots+|V_L(T_{\Delta})|
\end{align}

From equations~\ref{eq:internal-change-case12} and~\ref{eq:others-change-case12} we conclude the following
contradictory result:
\begin{align}
\label{eq:others-case12-final}
\sigma_{LL}(T) - \sigma_{LL}(\widetilde{T})= & \\
& \sigma_T(\mathcal{L}, V_L(T_{11})) - \sigma_{\widetilde{T}} (\widetilde{\mathcal{L}}, V_L(\widetilde{T}_2))+ \nonumber\\
& \sigma_T(\mathcal{L}, V_L(T_{2})) - \sigma_{\widetilde{T}} (\widetilde{\mathcal{L}}, V_L(\widetilde{T}_{11}))+ \nonumber\\
& \sigma_T(\mathcal{L}, V_L(T) - V_L(T_{11})\uplus V_L(T_{2})) - \sigma_{\widetilde{T}} (\widetilde{\mathcal{L}}, V_L(\widetilde{T}) - V_L(\widetilde{T}_{11})\uplus V_L(\widetilde{T}_{2})) & \nonumber\\
& > 0 \nonumber
\end{align}

\noindent
\paragraph{Case 2: $\forall w_1, w_2 \in V_L(T), |d_T(w_1,v) - d_T(w_2,v)| < 2$.}
Since based on the assumption of induction $T_1,\ldots,T_{\Delta} \in \mathfrak{T}(\Delta-1, \Delta)$, then
$T$ can be constructed from some tree $T_0 \in \mathfrak{T}(\mathcal{R}, \Delta)$ of order $|V(T_0)| = M_k(\mathcal{R}, \Delta)$,
by attaching $n - M_k(\mathcal{R}, \Delta)$ nodes to the leaf nodes of $T_0$.
Therefore, based on lemma~\ref{lemma:same-height-opt-if-in-family}, $T$ is optimal iff $T \in \mathfrak{T}(\Delta, \Delta)$.
\end{proof}
\begin{corollary}
\label{coral:opt-tree-layout-complete-graph}
Consider the complete graph $G$ of order $n = |V(G)|$. Tree $T$ is an optimal tree layout for $G$ iff $|V_L(T)| = n$ and $T \in \mathfrak{T}(3, 3)$.
\end{corollary}
\begin{exaxxx}
Let $G$ be a complete graph of order $n = 2 ^ l$ for some $ l > 1$. Based on the result of theorem~\ref{thr:opt-LL-tree}, a layout tree $T$ for $G$ (of order $2n - 1$) has minimum value $\sigma_{LL}$ (and accordingly is a solution for Min Layout Length) iff it is isomorphic to some tree with a structure similar to the tree in figure~\ref{fig:opt-layout-tree-for-complete}.
\begin{figure}[h]
        \centering
        \begin{subfigure}[b]{0.3\textwidth}
                \includegraphics[scale=.7]
                {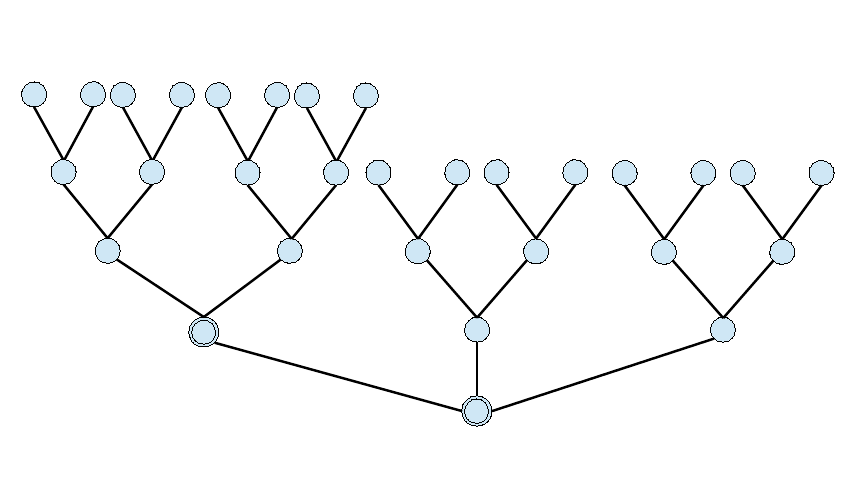}
        \end{subfigure}
        \qquad\qquad
        \begin{subfigure}[b]{0.3\textwidth}
                \includegraphics[scale=.7]
                {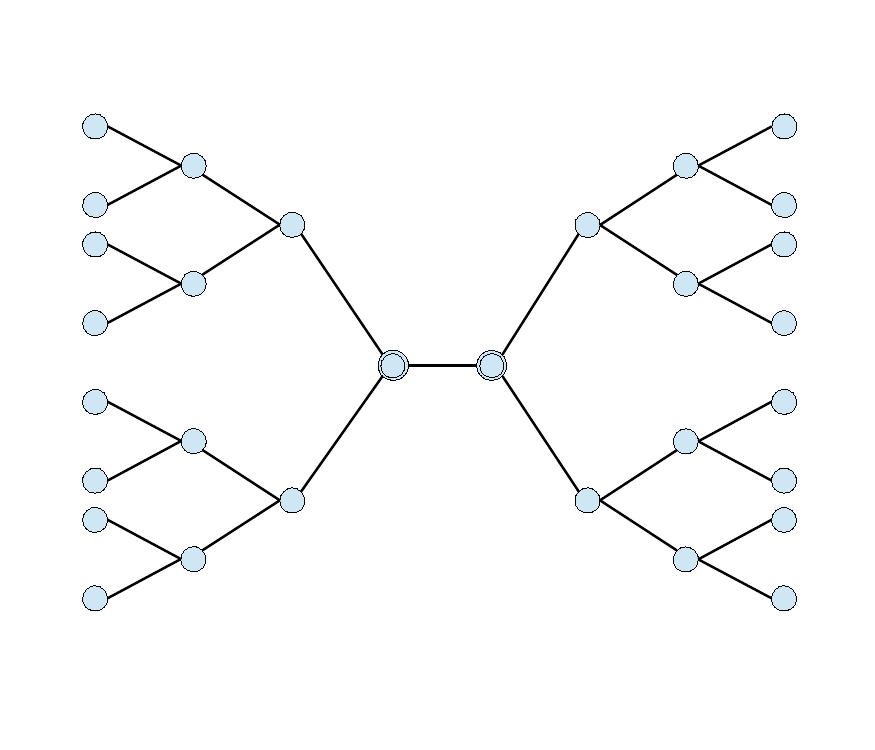}
        \end{subfigure}
        \caption{Two planar embeddings of an optimal layout tree $T$ of size $2n - 1$, corresponding to the complete graph $G$ with $n = 2^l$ vertices.}
    \label{fig:opt-layout-tree-for-complete}
\end{figure}
\end{exaxxx}

\subsection{Min Tree Length of Multi-Graphs}
\label{sec:multi-graph-optimal-layout}
Graph $G$ is a multi-graph if either it is a simple undirected graph, or it can be constructed from
a simple undirected graph by adding parallel edges. In our main result of this report
we show that the Min Tree Length problem is NP-hard for class of multi-graphs.
Finally we show that this result can be extended to the class of simple graphs.
\begin{definition}{Equal Size $4$-Clique Cover} Given graph $G$, Equal Size $4$-Clique Cover problem is the problem of partitioning $V(G)$ into
four disjoint subsets  $V_1, \ldots, V_4$ s.t. $G(V_i)$ is a clique of size $\dfrac{n}{4}$, for $1 \le i \le 4$ .
\end{definition}
\begin{lemma}
\label{lemma:equal-4-clique-cover}
Equal Size $4$-Clique Cover problem is NP-complete, even for the class of graphs of order $2^l$ vertices for some $l \in \mathbb{N}$.
\end{lemma}
\noindent
For proof you can refer to~\cite{}.

\begin{theorem}
\label{thr:multi-graph-tree-length-np-hard}
Min Tree Length problem is NP-hard for the class of multi-graphs.
\end{theorem}
\begin{proof}
The correctness of the theorem can be represented using a polynomial reduction form Equal Size $4$-Clique Cover.

\noindent
Consider an arbitrary graph $G$, as an instance input of Equal Size $4$-Clique Cover problem, where $|V(G)| = 2^l$ for some $l \in \mathbb{N}$.
Let $G^{\prime}$ be the multi-graph obtained from $G$ by introducing $M = m\times(2n-2)$ parallel edges
between every two vertices $u, v \in V(G)$.
Notice that every tree layout $T$ for graph $G$ has $2n-2$ edges where the congestion
of each edge is less than $m = |E(G)|$.
Considering graph $G$ as an vertex induced subgraph of $G^{\prime}$, then $G^{\prime} = G \uplus \widetilde{G}$, where $\widetilde{G}$
is another vertex induced subgraph of $G^{\prime}$. Subgraph $\widetilde{G}$ is complete multi-graph.
Hence for every tree layout $T$ and $\varphi: V(G) \rightarrow V_L(T)$ for $G^{\prime}$ we have:
\begin{align}
L(T, \varphi, G^{\prime}) = L(T,G) + L(T, \varphi, \widetilde{G})
\end{align}
\noindent
From corollary~\ref{coral:opt-tree-layout-complete-graph} we know that a layout tree $\widetilde{T}$ for $\widetilde{G}$
is optimal iff $\widetilde{T} \in \mathfrak{T}(3,3)$. Also for every layout tree $T \notin \mathfrak{T}(3,3)$, it is the case that
$L(T, _ , \widetilde{G}) \ge  L(\widetilde{T}, _ , \widetilde{G}) + M$.
On the other hand for $n > 2$ and every layout tree $T$ where $|V_L(T)| = n$, it is always the case that $L(T,G) < M$.

\noindent
Therefore, a layout tree $T$ for $G^{\prime}$ is optimal iff
$T$ and $\widetilde{T}$ are isomorphic.
In other words $T$ for $G^{\prime}$ is optimal iff $T \in \mathfrak{T}(3,3)$.
Hence the Min Tree Length problem for $G^{\prime}$ reduces to the problem of finding an optimal bijection $\varphi$ form
vertices of $G$ to leaf nodes of $\widetilde{T} \in \mathfrak{T}(3,3)$, such that the summation of
edge dilations for all $\Set{u,v} \in E(G)$ is minimized. Formally speaking:
\begin{align}
\argmin_{(T,\varphi)} L(T, \varphi, G^{\prime}) = (\widetilde{T}, \argmin_{\varphi} L(\widetilde{T}, \varphi, G) )
\end{align}

\noindent
Let $\overline{G}$ denote the complement of graph $G$. One can easily check that:
\begin{align}
L(\widetilde{T}, \varphi, G) = L(\widetilde{T}, \varphi, K_n) - L(\widetilde{T}, \varphi, G)
\end{align}
\noindent
Where $K_n$ is a complete graph of size $n$ .Therefore:
\begin{align}
\argmin_{\varphi} L(\widetilde{T}, \varphi, G) = \argmax_{\varphi} L(\widetilde{T}, \varphi, G) = \argmax_{\varphi} \sum_{\Set{u,v} \in E(\overline{G})} \lambda(uv,\widetilde{T},\varphi,G)
\end{align}

\noindent
Also based on the structure of $\widetilde{T}$, $\exists c_1, c_2 \in \mathcal{C}_{\widetilde{T}}$ (in figure~\ref{fig:opt-layout-tree-for-complete} shown by double lined circles). Since $n = 2^l$ for $l \in \mathbb{N}$, removing central nodes $c_1$ and $c_2$ partition leaf nodes of $\widetilde{T}$ into 4 equally sized
partitions $L_1, \ldots,L_4$.

\noindent
Finally, we can deduce that the original graph $G$ can be partitioned into 4 complete sub-graphs of size $\dfrac{n}{4}$, iff there exist
bijection $\varphi$ such that $\forall \Set{u,v} \in E(\overline{G}), \varphi(u) \in L_i \wedge \varphi(v) \in L_j$ for $1 \le i \neq j \le 4$.
In other words, graph $G$ can be partitioned into 4 complete sub-graphs $G_1,\ldots,G_4$ of size $\dfrac{n}{4}$,
iff $(\widetilde{T}, \varphi)$ is a solution for Min Tree Length of $G^{\prime}$, where
$\varphi$ maps vertices of $G_i$ to leaf nodes of $L_i$ for $1 \le i \neq j \le 4$.
Which infers the NP-hardness of Min Tree Length problem for the multi-graphs.
\end{proof}

\subsection{Min Tree Length of Simple Graphs}
\label{sec:simple-graph-optimal-layout}
Finite graph $G^{\prime}$ is simple, if for $u \neq v \in V(G^{\prime})$ there is at most one edge $\Set{u,v} \in E(G^{\prime})$
and $\forall u \in V(G), \Set{u,u} \notin E(G)$. Consider complete multi-graph $G$, where
every two vertices $u$ and $v$ are connected via $l$ parallel edges.
Having multi-graph $G$ one can obtain a simple graph $G^{\prime}$ by subdividing very edge $\Set{u,v} \in E(G)$
and introducing a new vertex $x$ of degree $2$.

\noindent
Consider a tree layout $T^{\prime}$ and bijective mapping $\varphi^{\prime}: V(G^{\prime}) \rightarrow V_L(T^{\prime})$ for the simple graph $G^{\prime}$.
For every $x \in V(G^{\prime})$, $\varphi^{\prime}(x) \in V_L(T^{\prime})$ is directly connected to some internal node $w \in V_I(T^{\prime})$.
Removing $w$ results in the set of three subtree $\Set{T^{\prime}_{w,1}, \varphi^{\prime}(u), T^{\prime}_{w,2}}$.
For every $x \in V(G^{\prime})$, by $\mathcal{T}^{\prime}_1(x)$ and $\mathcal{T}^{\prime}_2(x)$
we refer respectively to subtree $T^{\prime}_{w,1}$ and $T^{\prime}_{w,2}$, where $w$ is the direct neighbour of $\varphi^{\prime}(x)$ in $T^{\prime}$.
It is easy to see that $\sigma(\Set{\varphi^{\prime}(x),w}, T^{\prime}, \varphi^{\prime}, G^{\prime})$ is equal to the degree of $x$ in $G$, and also if
$x$ is of degree $2$ then the congestion of the two edges connecting $w$ to $T^{\prime}_{w,1}$ and $T^{\prime}_{w,2}$ is equal iff
$\varphi^{\prime}(u) \in V(T^{\prime}_{w,1})$ and $\varphi^{\prime}(v) \in T^{\prime}_{w,2}$, where $u$ and $v$ are the direct neighbors of $x$ in $G$.

\noindent
Consider tree layout $T$ and bijective mapping $\varphi: V(G) \rightarrow V_L(T)$ for the multi-graph $G$.
Starting from tree $T$, we can constructed a new tree $T^{\prime}$ for the simple graph $G^{\prime}$ by subdividing some edge of $T$
and introducing sub-tree layouts containing only vertices of degree $2$.
In other words, $T^{\prime}$ is constructed from $T$ by introducing $m=|E(G)|$ internal node and $m$ leaf nodes (each corresponding to one vertex of $G^{\prime}$ of degree $2$).
Figures ~\ref{fig:complete-multi-graph} and~\ref{fig:complete-multi-graph-tree} respectively depict the
multi-graph $G$ and its corresponding tree layout $T$, while
in figures~\ref{fig:complete-simple-graph} and~\ref{fig:complete-simple-graph-tree} you can see simple graph $G^{\prime}$ obtained from $G$ and one possible tree layout $T^{\prime}$ for simple graph $G^{\prime}$ constructed from
$T$.
\begin{figure}[h]
        \centering
        \begin{subfigure}[b]{0.3\textwidth}
               \includegraphics[scale=.7]
                {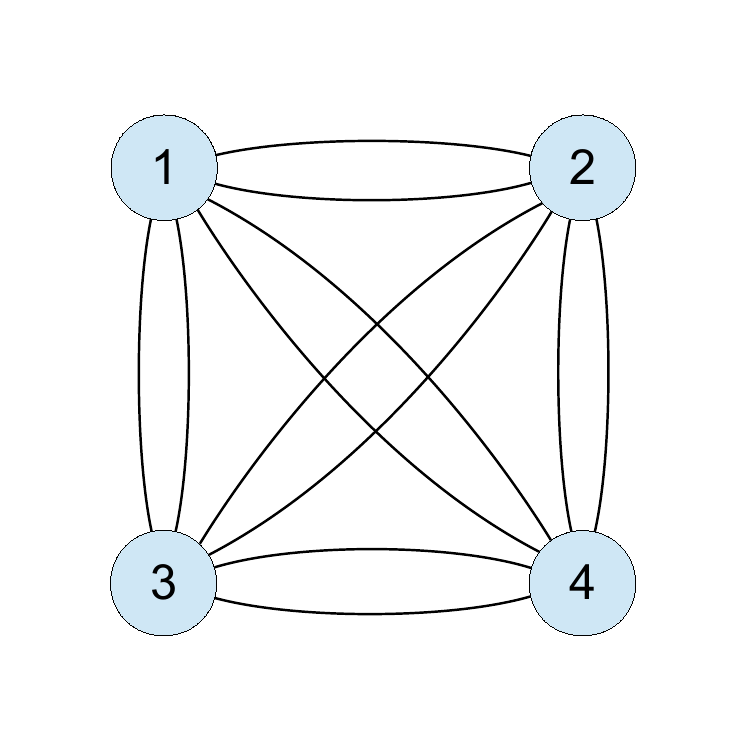}
                \caption{Complete multi-graph $G$ where every two vertices are connected via $l = 2$ parallel edges.}
                 \label{fig:complete-multi-graph}
        \end{subfigure}
        \qquad
        \begin{subfigure}[b]{0.3\textwidth}
                \includegraphics[scale=.7]
                {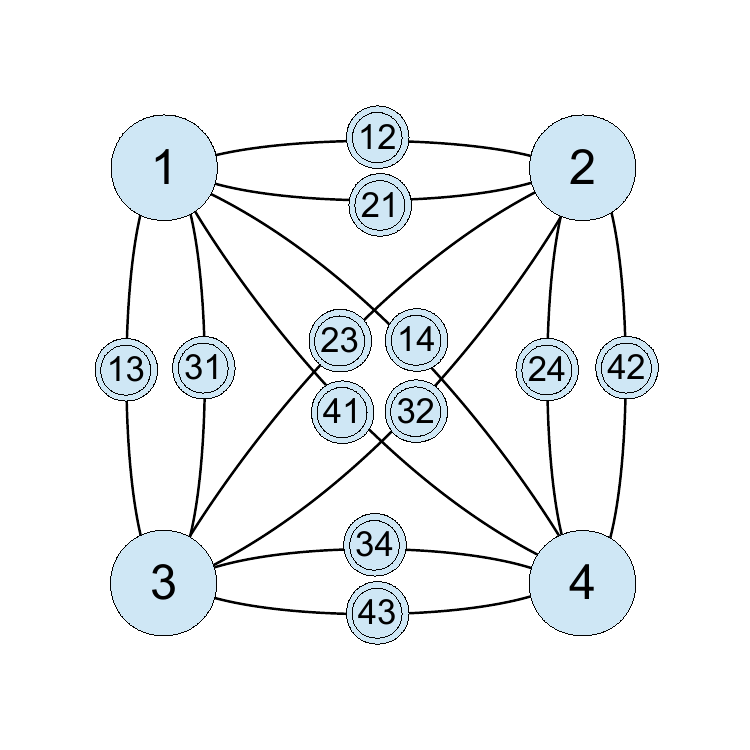}
                \caption{Simple graph $G^{\prime}$ obtained from $G$.}
                \label{fig:complete-simple-graph}
        \end{subfigure}
        \caption{Simple graph $G^{\prime}$ obtained from $G$ by subdividing every edge of $G$ and introducing a vertex of degree $2$.}
\end{figure}
\begin{figure}[h]
        \centering
        \begin{subfigure}[b]{0.3\textwidth}
                \includegraphics[scale=.7]
                {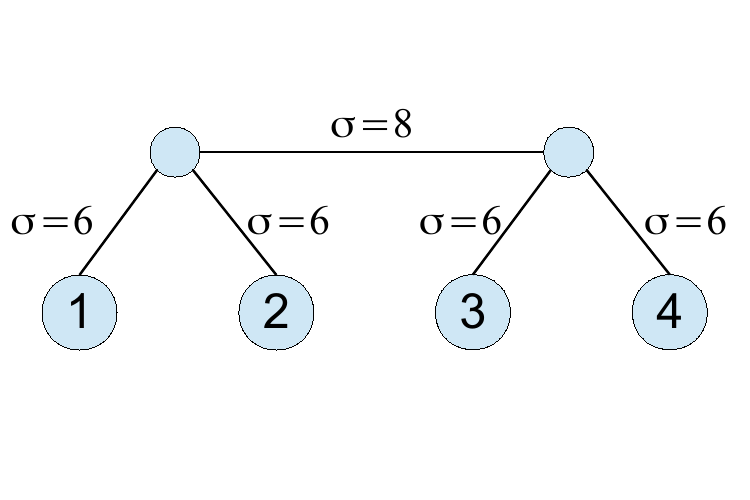}
                \caption{An arbitrary tree layout $T$ for $G$.}
                \label{fig:complete-multi-graph-tree}
        \end{subfigure}
        \qquad\\
        \begin{subfigure}[b]{0.6\textwidth}
                \includegraphics[scale=.7]
                {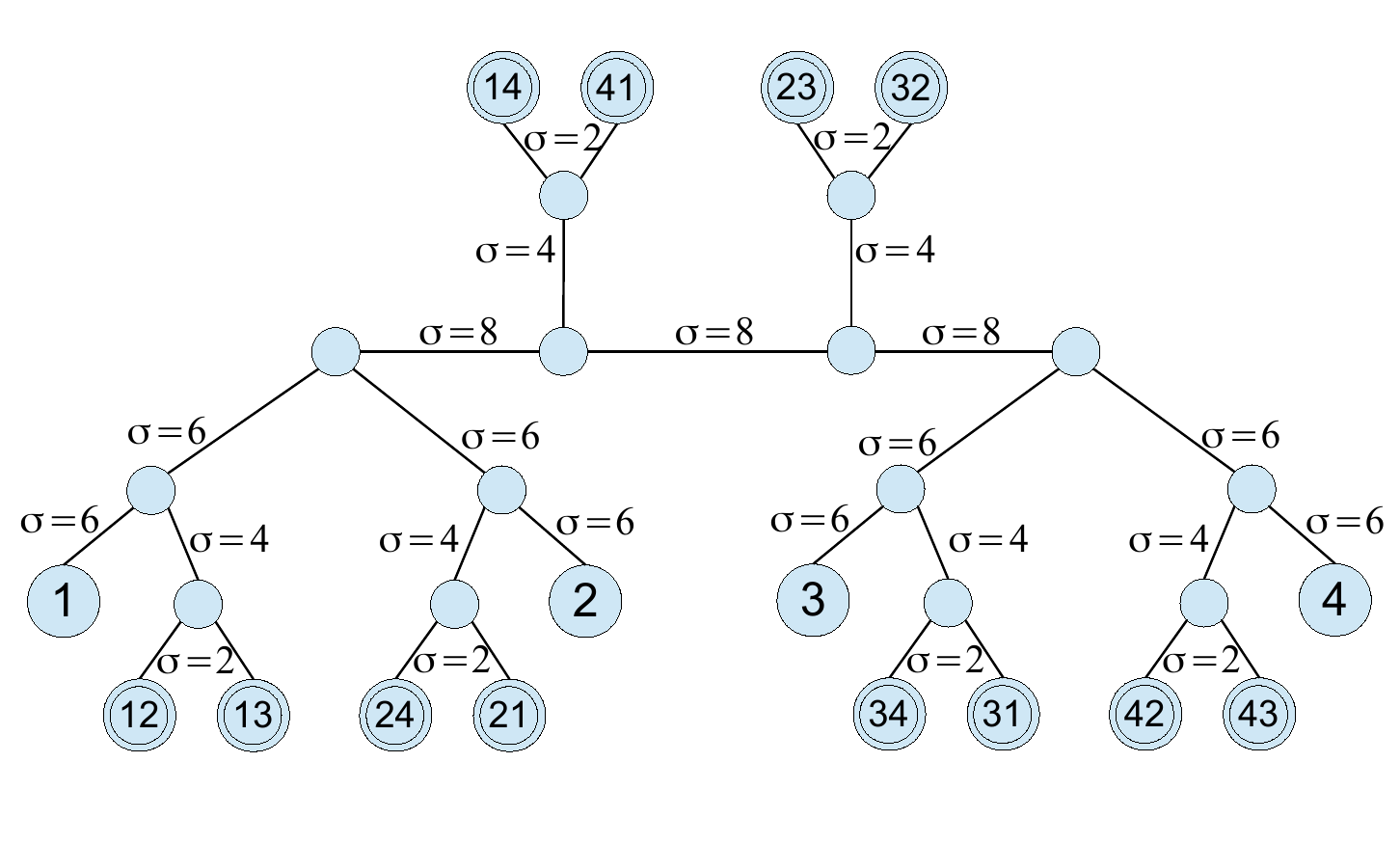}
                \caption{One possible tree layout $T^{\prime}$ for $G^{\prime}$, constructed from $T$. As you can see,
                $T^{\prime}$ is obtained by subdividing some edges of tree $T$ (possibly several times) and introducing sub-tree layouts containing only vertices of degree $2$.}
                \label{fig:complete-simple-graph-tree}
        \end{subfigure}
        \caption{Tree layout $T^{\prime}$ for simple graph $G^{\prime}$, constructed from initial tree layout $T$. The congestion of
        every edge is include by that edge.}
\end{figure}
\begin{lemma}
\label{lemma-simple-graph-optimal}
Let $G$ be a complete multi-graph where $\forall u,v \in V(G)$ there exist $l$ parallel edge $\Set{u,v} \in E(G)$.
Let $T$ and $\varphi$ be arbitrary tree layout and the corresponding mapping for $G$. Also assume $G^{\prime}$ is the simple graph obtained by subdividing every edge of $G$.
We show the class of all possible layout trees for $G^{\prime}$, constructed from $T$ by $\mathfrak{F}(T)$.
Consider tree layout $T^{\prime} \in \mathfrak{F}(T)$, where $\forall T^{\prime\prime} \in \mathfrak{F}(T), LA(T^{\prime}, G) \le LA(T^{\prime\prime}, G)$.
Then it is the case that $T^{\prime}$ is constructed by subdividing only edges of $T_0$ that each one is adjacent to a leaf node (which we call external edges).
\end{lemma}
\begin{proof}
We know that for every edge $e \in E(T), \sigma(e, T, G) \ge l\times(n-1)$ where $n = V(G)$ and equality holds only if
$e$ is adjacent to a leaf node $w \in V_L(T)$.
Now consider an arbitrary $T^{\prime\prime} \in \mathfrak{F}(T)$, obtained by subdividing at least one \emph{internal} edge $e \in E(T)$\footnote{
Edge $e \in E(T)$ is external if it is adjacent to a leaf node of $T$, and internal otherwise. Let $E_I(T)$ and $E_E(T)$ respectively represent the
set of internal and external edges of $T$}.
It is easy to check that $\forall e\in E_I(T), \sigma(e, T, G) \ge l \times 2 \times (n-2)$. As appose to tree layout $T^{\prime\prime}$ for $G^{\prime}$, we suggest tree layout $T^{\prime}$ constructed by relocating the subtree (possibly more than one subtree) that subdivides $e$ (and hence removing all the subdivisions of $e$) and
creating a new subdivision (possibly more than one subdivision) in a an external edge $\Set{w,w^{\prime}}$ (assume $w \in V_L(T)$).
We choose an external edge $\Set{w,w^{\prime}}$ such that for the every leaf node $\overline{x}$ of the relocated subtree,
$\Set{\varphi^{-1}(\overline{x}), \varphi^{-1}(w)} \in E(G^{\prime})$.
Hence based on the construction of $T^{\prime}$ form $T^{\prime\prime}$,
the following equation holds, which in turn evidences the correctness of the lemma.
\begin{align*}
LA(T^{\prime\prime}, G^{\prime}) - LA(T^{\prime}, G^{\prime}) > ( l \times 2 \times (n-2)) - (l \times (n-1)) > 0
\end{align*}
\end{proof}
\noindent
From the result of lemma~\ref{lemma-simple-graph-optimal} one can infer the fact that
given layout tree $T$ and mapping $\varphi$ for complete multi-graph $G$, an optimal tree layout (based on $T$) for simple graph $G^{\prime}$
can be constructed by subdividing only external edges of $T$
and introducing sub-tree layouts corresponding to the new vertices of degree $2$. But it does not
provide any information regarding the exact structure of the optimal tree.
In what follows and without providing all the details of the proof,
we present the structure of the optimal layout tree for $G^{\prime}$, constructed from $T$.
Note that for the sake of the main theorem in this subsection, we do not need to know the 
exact structure of the three layout with minimum tree length.

\noindent
The simple graph $G^{\prime}$ contains exactly $l\times \dfrac{n(n-1)}{2}$ vertices of degree $2$.
Every vertex $v \in G^{\prime}$ of degree $l\times (n-1)$ is directly connected to $l\times (n-1)$ vertices
of degree $2$.
The optimal tree layout $T^{\prime}$ obtains from $T$, by subdividing every external edge $\Set{w,w^{\prime}} \in E_E(T)$ (where $w \in V_L(T)$)
exactly once with a sub-tree layout containing $\dfrac{l}{2}\times (n-1)$  leaf nodes\footnote{For the sake of the main theorem in this section, we assume $l$ is an even integer.}. Every leaf node of this subtree
correspond to one vertex $x \in V(G^{\prime})$ with degree $2$ where $x$ is directly connected to $w$.

\noindent
Let $\mathcal{T}$ be an subtree of tree layout $T^{\prime}$ where its leaf nodes correspond to only vertices of degree $2$ in $G^{\prime}$. Since no two distinct vertices $x,y \in V(G^{\prime})$ of degree $2$ are neighbors,
the summation of congestions of edges of a subtree $\mathcal{T}$ (with a fixed number of nodes) is minimum only when
$\mathcal{T}$ is a complete rooted binary tree with no node of degree $2$. The root node of $\mathcal{T}$ is directly connected to the
internal node subdividing edge $\Set{w,w^{\prime}} \in E_E(T)$.
In the suggested optimal tree $T^{\prime}$ constructed from $T$, for every newly introduced edge $e$,
we have $\sigma(T^{\prime}, G^{\prime}) \le l \times (n-1)$.

\noindent
Figure~\ref{fig:complete-simple-graph-opt-tree} depicts the structure of an optimal layout tree for $G^{\prime}$
constructed from the initial layout tree $T$.
\begin{figure}[h]
        \centering
        \includegraphics[scale=.7]
        {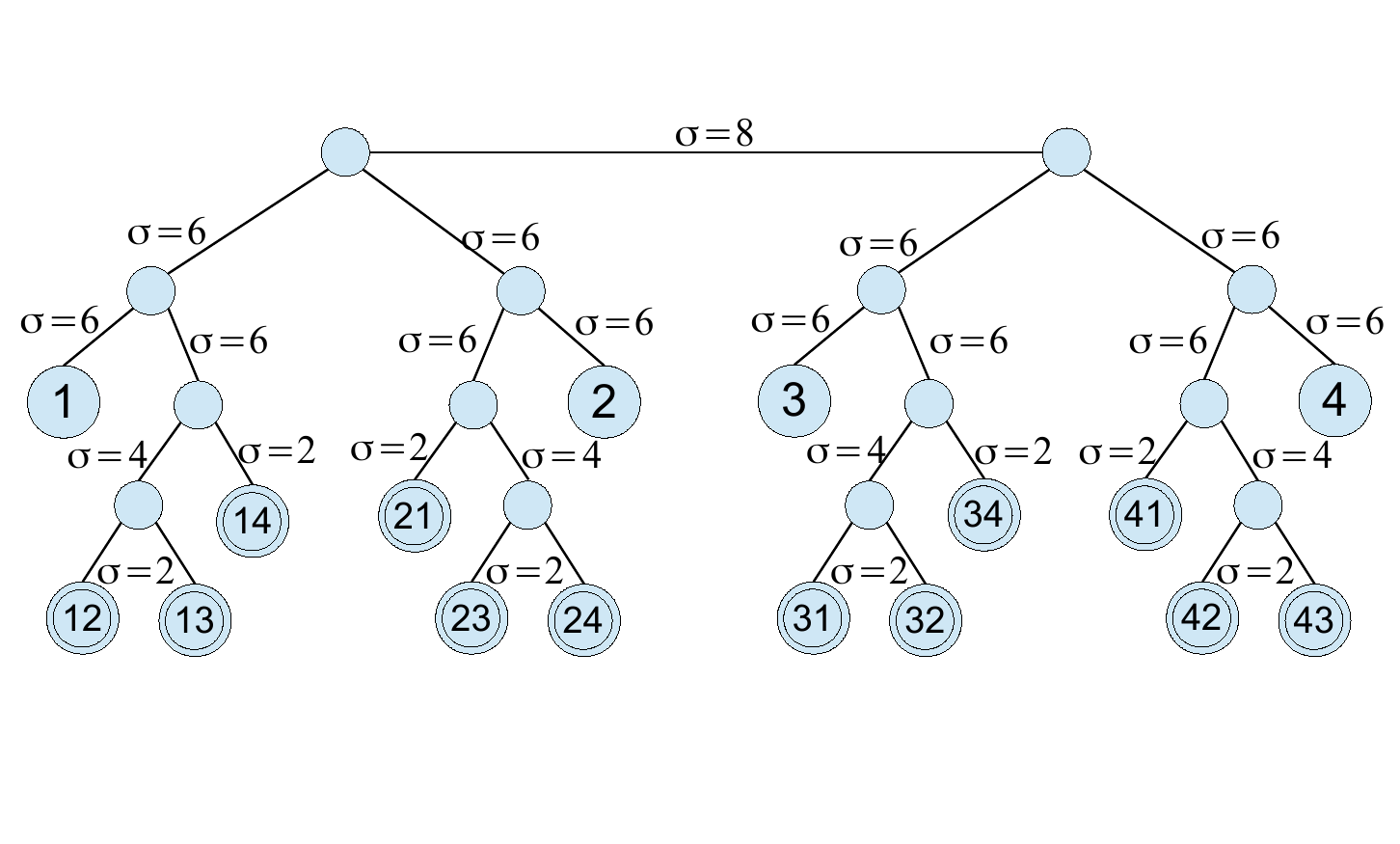}
        \caption{Optimal tree layout $T^{\prime}$ for $G^{\prime}$, constructed from $T$.}
        \label{fig:complete-simple-graph-opt-tree}
\end{figure}
\noindent
Accordingly the value of optimal layout $T^{\prime}$ for $G^{\prime}$, based on the initial tree layout $T$ for $G$, is equal to:
\begin{align*}
    LA(T^{\prime}, G^{\prime}) = LA(T, G) + n\times(\mathbb{T} + l\times(n-1))
\end{align*}
\noindent
Where the constant $\mathbb{T}$ is the summation of all edges' congestion of an optimal subtree, containing $\dfrac{l}{2}\times (n-1)$ leaf nodes, such that every leaf node correspond to a vertex of degree $2$ in $G^{\prime}$. Also constant $l\times(n-1)$ is the congestion of every external edge in $T$.
As you can see, the constant part of this equation does not depend on
the structure of the initial tree layout $T$. Hence consider two tree layouts $T^{\prime}$ and $T^{\prime\prime}$ for $G^{\prime}$,
respectively optimally obtained from tree layouts $T^{\prime}_0$ and $T^{\prime\prime}_0$ for complete multi-graph $G$.
Then $LA(T^{\prime}, G^{\prime}) < LA(T^{\prime\prime}, G^{\prime})$ iff $LA(T^{\prime}_0, G) < LA(T^{\prime\prime}_0, G)$.
\begin{corollary}
\label{col:complete-simple-opt-tree}
Let $T \in \mathfrak{T}(3,3)$ and $\varphi: V(G) \rightarrow V_L(T)$ be the solution of the Min Tree Length problem for the complete multi-graph $G$,
where every two distinct vertices are connected via $l$ parallel edges.
Assume $G^{\prime}$ is the simple graph obtained from $G$ by subdividing every edge with a vertex of degree $2$.
Also, let layout $T^{\prime}$ and bijective mapping $\varphi^{\prime}: V(G^{\prime}) \rightarrow V_L(T^{\prime})$ be the optimal solution of Min Tree Length problem for graph $G^{\prime}$. Then it is the case that:
\begin{itemize}
  \item $\forall v \in V(G^{\prime}) \cap V(G), \varphi^{\prime}(v) = \varphi(v)$
  \item $T^{\prime}$ is constructed from $T$ by subdividing every
        external edge $\Set{w,w^{\prime}}$ of $T$ using sub-tree layout $\mathcal{T}_w$, containing $\dfrac{l}{2}\times (n-1)$  leaf nodes, where
  \item for every leaf node $\overline{x} \in V_L(\mathcal{T}_w)$, ${\varphi^{\prime}}^{-1}(\overline{x}) = x \in V(G^{\prime})$, where $x$ is directly connected to ${\varphi^{\prime}}^{-1}(w)$ in $G^{\prime}$.
\end{itemize}
\end{corollary}
\noindent
Facilitating the result of corollary~\ref{col:complete-simple-opt-tree}, we conclude this section by showing that
the Min Tree Length problem stays NP-hard even for the class of simple graphs.
\begin{theorem}
\label{thr:simple-graph-tree-length-np-hard}
Min Tree Length problem is NP-hard for the class of simple graphs.
\end{theorem}
\begin{proof}
This theorem can be proven using a similar approach as we used in the proof
of theorem~\ref{thr:multi-graph-tree-length-np-hard} by a reduction form Equal Size $4$-Clique Cover problem.

\noindent
Hence given graph $G$, as an instance input of Equal Size $4$-Clique Cover problem, we construct
multi-graph, by introducing $M = m\times(2n-2)$ parallel edges
between every two vertices $u, v \in V(G)$.
In the next step we obtain a simple graph $G^{\prime}$ by subdividing every newly introduced edges.

\noindent
Considering graph $G$ as an vertex induced subgraph of $G^{\prime}$, then $G^{\prime} = G \uplus \widetilde{G}$, where $\widetilde{G}$
is also a vertex induced subgraph of $G^{\prime}$. $\widetilde{G}$ is the simple graph obtained from a complete multi-graph $G_0$, by subdividing every edge.
Hence for every tree layout $T$ and $\varphi: V(G) \rightarrow V_L(T)$ for $G^{\prime}$ we have:
\begin{align}
L(T, \varphi, G^{\prime}) = L(T,\varphi^{\prime}, G) + L(T, \varphi, \widetilde{G})
\end{align}
\noindent
Where $\varphi^{\prime}$ is partially defined from $\varphi$, in other words, $\varphi^{\prime}(v) = \varphi(v)$ for every $v \in V(G)$.

\noindent
From corollary~\ref{col:complete-simple-opt-tree} we know that a layout tree $\widetilde{T}$ for $\widetilde{G}$
is optimal iff $\widetilde{T}$ is optimally constructed from a tree layout $T_0 \in \mathfrak{T}(3,3)$ for $G_0$. Also for every layout tree $T^{\prime\prime}$ optimally constructed from some tree layout $T_i \notin \mathfrak{T}(3,3)$, it is the case that
$L(T^{\prime\prime}, _ , \widetilde{G}) \ge  L(\widetilde{T}, _ , \widetilde{G}) + M$.
On the other hand for $n > 2$ and every layout tree $T$ where $|V_L(T)| = n$, it is always the case that $L(T,_,G) < M$.

\noindent
Optimal tree $\widetilde{T}$ for $G^{\prime}$ has a similar structure to the structure of $T_0 \in \mathfrak{T}(3,3)$, in the sense that:
\begin{itemize}
  \item the leaf nodes can be partitioned into the 4 subtrees $\widetilde{T}_1,\ldots, \widetilde{T}_4$ of the same size and isomorphic structure,
  \item every subtree $\widetilde{T}_i$ contains $\dfrac{n}{4}$ leaf nodes, corresponding to the vertices of $G$, where
  \item $\forall w_1, w_2 \in \widetilde{T}_i, w_2 \in \widetilde{T}_j$ for $i \neq j$, where $\varphi^{-1}(w_1), \varphi^{-1}(w_2), \varphi^{-1}(w_3) \in V(G)$, it is the case that $d_{\widetilde{T}}(w_1,w_2) < d_{\widetilde{T}}(w_1,w_3)$. In other words,
      the distance of every two leaf nodes $w_1, w_2$ (corresponding to vertices of $G$) in the same subtree $\widetilde{T}_i$ is less than the distance of every two leaf nodes that belong to two distinct subtrees $\widetilde{T}_i$ and $\widetilde{T}_j$.
\end{itemize}
\noindent
Hence the Min Tree Length problem for $G^{\prime}$ reduces to the problem of finding an optimal bijection form
vertices of $G$ to the leaf nodes of $\widetilde{T}$ (that correspond to the vertices of $G$), such that the summation of
edge dilations for all $\Set{u,v} \in E(G)$ is minimized.
Therefore, similar the proof of theorem~\ref{lemma:equal-4-clique-cover} and omitting the details, it can be inferred that $G$ is an positive
instance of $4$-Clique Cover problem, iff using the optimal tree layout $\widetilde{T}$ for $G^{\prime}$, vertices of $G$
can be partitioned into $4$ complete graphs of size $\dfrac{n}{4}$, which indicates the NP-hardness of Min Tree Length problem for
the class of simple graphs.
\end{proof}

\section{Layout Tree Problem in Relation with Graph Reassembling}
\label{sect:reassembling}
  In this section we study the relation between tree layout problem and \emph{graph reassembling} problem as defined in~\cite{kfoury2015efficient}.
Graph reassembling problem plays a key role in the efficiency of main programs in earlier work on a domain-specific language (DSL) for the design of flow networks~\cite{BestKfoury:dsl11,Kfoury:sblp11,kfoury2013different}.

Consider a simple graph graph $G = (V,E)$ (not necessarily connected), partitioned into the set of $|V|$ one-vertex components
by cutting every edge into into two halves. Reassembling of the graph $G$ corresponds to the problem of finding
the sequence of edge reconnections $\Theta$ that minimizes two measures that depend on the edge-boundary degrees of assembled components
in the intermediate steps of reassembling $G$. The first step of $\Theta$ initiates with the set of one-vertex components, and
the final step results in the initial graph $G$.
The optimization goal of the graph reassembling can be either minimizing the maximum edge-boundary degree encountered during the
reassembling process, which is called the $\alpha$-measure of the reassembling, or the sum of all edge-boundary
degrees, denoted by $\beta$-measure. 

\noindent
The reassembling sequence $\Theta$ for graph $G$ correspond to a unique binary tree $\B$ (called \emph{binary reassembling tree}), where
the set of leaf nodes is bijectively related to the set of one-vertex components and the root node correspond to the reassembled graph $G$.
Every internal node $w \in V_I(\B)$, as the root of the a subtree $\B_w$, correspond to the vertex induced sub-graph of $G$ comprising
the set of vertices represented by leaf nodes of $\B_w$.

\noindent
Hence, the $\alpha$-optimal reassembling problem for graph $G$ is the problem of finding
a \emph{rooted} binary tree $\B$ and a bijective mapping from vertices of $G$ to leaf node of $\B$,
such that the maximum edge congestion of $\B$ is minimized.
Similarly the $\beta$-optimal reassembling problem is the problem of finding
a \emph{rooted} binary tree $\B$ and a bijective mapping from vertices of $G$ to leaf node of $\B$,
where tree length of $\B$ is minimized. It is easy to see that all the result for the minimum tree layout congestion problem
can be directly inferred for the case where the underlying tree is rooted. On the other hand
the same statement can not immediately be inferred for minimum tree layout length problem.
Therefore in this section we study the Min Tree Length problem where the host graph is a rooted
binary tree (Min Rooted Tree Length problem for short).
\begin{lemma}
\label{lemma:min-tree-length-degree-1}
Min Tree Length problem is NP-hard for the class of graphs $\G_{\nabla=1}$ where for every member $G \in \G_{\nabla=1}$, $G$ is connected and exists $v \in V(G)$ of degree $1$ (\emph{i.e.} $\G_{\nabla=1}$ is a the class of graphs with min degree $\nabla=1$).
\end{lemma}
\begin{proof}
 One dimidiate result of the methods that are used in proofs of theorems~\ref{thr:multi-graph-tree-length-np-hard} and~\ref{thr:simple-graph-tree-length-np-hard} is that the Min Tree Length
 problem stays NP-hard even for the class of \emph{congested graphs}. A graph $G$ with minimum vertex degree $\nabla$
 is congested if for every tree layout $T$ for $G$ it is the case $\sigma(e,T,_,G) \ge \nabla$ for every edge $e \in E(T)$.
 \footnote{The graphs that are used in both proof are clearly congested graphs.}

\noindent
Consider congested graph $G$ with min degree $\nabla$. if $\nabla = 1$ we are done, otherwise
let $v \in V(G)$ be a vertex with degree $\nabla$ and $G^{\prime}$ be
the graph constructed by augmenting $G$ with a new vertex $u$ and edge $\Set{u,v}$.
Also let tree layout $T^{\prime}$ and mapping $\varphi^{\prime}$  be a solution for
Min Tree Length problem of $G^{\prime}$ where $\Set{\varphi^{\prime}(v), w}$ is edge incident to leaf node $\varphi^{\prime}(v)$ (with congestion $\nabla$). Is not hard to verify that it must be the case that $\Set{\varphi^{\prime}(u), w} \in E_E(T^{\prime})$.

\noindent
Let $T$ be a layout tree for $G$, obtained from $T^{\prime}$ after removing vertices $\varphi^{\prime}(u)$ and $w$ (and their incident edges)
and introducing edge $\Set{\varphi^{\prime}(v), w^{\prime}}$. Where $w^{\prime}$ is the third neighbor of $w$ in $T^{\prime}$.
Hence the following equality holds.
\begin{align*}
LA(T^{\prime},\varphi^{\prime}, G^{\prime}) = LA(T,\varphi, G) + 2
\end{align*}
\noindent
Where $\varphi(v) = \varphi^{\prime}(v)$ for every $v \in V(G)$.
\begin{claimxxx}
Tree layout $T$ and along side with the bijective mapping $\varphi$ is a solution for Min Tree Length problem of $G$.
\end{claimxxx}
Assume there exist Tree layout $\widetilde{T}$ and bijective mapping $\widetilde{\varphi}$ such that $LA(\widetilde{T},\widetilde{\varphi}, G) < LA(T,\varphi, G)$. Hence one can construct a tree layout $\widetilde{T}^{\prime}$ and bijective mapping $\widetilde{\varphi}^{\prime}$ for $G^{\prime}$ by subdividing the edge incident to leaf node $\widetilde{\varphi}(v)$ and introducing internal node $w$ which is directly connected
to leaf node $\widetilde{\varphi}^{\prime}(u)$. Therefore:
\begin{align*}
LA(\widetilde{T}^{\prime},\widetilde{\varphi}^{\prime}, G^{\prime}) = LA(\widetilde{T},\widetilde{\varphi}, G) + 2 < LA(T^{\prime},\varphi^{\prime}, G^{\prime})
\end{align*}
\noindent
Which contradicts the assumption that $T^{\prime}$ and $\varphi^{\prime}$  are a solution for
Min Tree Length problem of $G^{\prime}$.
\end{proof}
\noindent
Using the result of lemma~\ref{lemma:min-tree-length-degree-1} it can be shown that Min Rooted Tree Length is NP-hard
for the class of non necessary connected graphs. Later on we extend the this result to the class of connected graphs.
\begin{theorem}
\label{lemma:disjoint-component-min-rooted-tree-length}
Consider the finite undirected graph $G$ (non necessarily connected). The problem of finding a bijective mapping from
the vertices of $G$ so some rooted binary tree $T$ with minimum length is not polynomially solvable unless P=NP.
\end{theorem}
\begin{proof}
We proof this theorem using the NP-hardness of Min Tree Length problem for the class of $\G_{\nabla=1}$.
Hence given a graph $G \in \G_{\nabla=1}$, we obtain a disconnected graph $\overline{G} = G \uplus G_0$ where $G_0 = (\Set{v}, \Set{})$ is a one-vertex
graph.
Assume rooted binary tree $\overline{T}$ and bijective mapping $\overline{\varphi}$ are the solution
for the Min Rooted Tree Length problem of the augmented graph $\overline{G}$.
\begin{claimxxx}
\label{claim:disjoint-component-min-rooted-tree-length}
\begin{enumerate}[label=(\Roman*)]
  \item $w_1 = \overline{\varphi} (v)$ is directly connected to the root node $r$ of $\overline{T}$, and
  \item let $w_2$ be the second direct neighbor of the root of $\overline{T}$. Node $w_2$ subdivides an edge (or equivalently, is connected to two edges) with congestion $1$. Figure~\ref{fig:disjoint-rooted-tree-length-degree-1} depicts the claimed structure of $\overline{T}$.
\end{enumerate}
\begin{figure}[h]
        \centering
        \includegraphics[scale=.7]
        {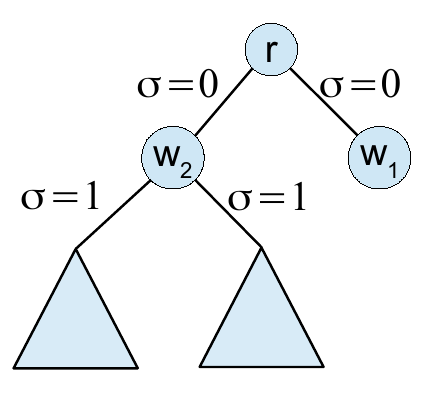}
        \caption{The structure of rooted tree layout for augmented graph $\overline{G}$.}
        \label{fig:disjoint-rooted-tree-length-degree-1}
\end{figure}
\noindent
\end{claimxxx}
\paragraph{Proof of I.}
Assume the opposite holds. Therefore $w_1 = \overline{\varphi} (v)$ is directly connected to an internal node of $w^{\prime} \in V_I( \overline{T})$, where $w^{\prime}$ is incident to two edge $e_1, e_2 \in E(\overline{T})$ such that $\sigma(e_1, \overline{T}, \overline{\varphi}, G) = \sigma(e_2, \overline{T}, \overline{\varphi}, G) > 0$. We construct an alternative rooted tree layout $\overline{T}^{\prime}$ from $\overline{T}$ as described in following.
We remove edge $\Set{w_1, w^{\prime}}$ and $w^{\prime}$ and introducing an new edge, replacing $e_1$ and $e_2$. Also we introduce a new root node $r^{\prime}$ and connect it to $w_1$ and $r$. Based on the construction process of $\overline{T}^{\prime}$, it is the case that $LA(\overline{T}, \overline{G}) - LA(\overline{T}^{\prime}, \overline{G}) = \sigma(e_2, \overline{T}, \overline{\varphi}, \overline{G}) > 0$, which contradicts the assumption of optimality of $\overline{T}$.
\paragraph{Proof of II.}
Clearly exists edge $e \in E(\overline{T})$ such that $\sigma(e, \overline{T}, \overline{\varphi}, \overline{G}) = 1$.
Hence assuming that $w_2$ subdivides an edge with congestion greater that $1$, similar to the approach in the proof of II, one can obtain an alternative rooted tree layout
$\overline{T}^{\prime}$ with the contradictory property $LA(\overline{T}^{\prime}, \overline{G}) < LA(\overline{T}, \overline{G})$.
\begin{claimxxx}
Graph $G \in \G_{\nabla=1}$ has tree length less than $k$, iff the augmented graph $\overline{G}$
has a rooted tree layout with tree length less than $k+1$.
More specifically given an rooted tree layout $\overline{T}$ and bijective mapping $\overline{\varphi}$ for augmented tree $\overline{G}$ (with the structure presented in claim~\ref{claim:disjoint-component-min-rooted-tree-length} and figure~\ref{fig:disjoint-rooted-tree-length-degree-1}),
removing nodes $r$, $w_1$ and $w_2$ and joining the two edges of congestion $1$ incident to $w_2$,
obtains an optimal tree layout $T$ and mapping $\varphi$ for graph $G$ (where $\varphi(u) = \overline{\varphi}(u)$ for every $u \in V(G)$).
\end{claimxxx}
\paragraph{Proof.}
It is easy to see that $LA(\overline{T}, \overline{G}) = LA(T, G) + 1$.
Now Assume there exist a tree layout $T^{\prime}$ and bijective mapping $\varphi^{\prime}$ for $G$
such that $LA(T^{\prime},\varphi^{\prime}, G) < LA(T,\varphi, G)$. On the other hand for every tree layout $T^{\prime}$ for $G$, exists $e \in E(T^{\prime})$ where
$\sigma(e, T^{\prime},_, G) = 1$. From $T^{\prime}$ (and $\varphi^{\prime}$), an alternative rooted tree layout $\overline{T}^{\prime}$ (and bijective mapping $\overline{\varphi}^{\prime}$) can be obtained by introducing three new nodes,
node $r$ designated as the root of $\overline{T}^{\prime}$, leaf node $w_1$, directly connected to $r$ (where $\overline{\varphi}^{\prime}(w_1) = v$)
and internal node $w_2$ directly connected to $r$ which subdivides edge $e$.
It can be verified that $LA(\overline{T}^{\prime},\overline{\varphi}^{\prime}, \overline{G}) = LA(T^{\prime}, \varphi^{\prime}, G) + 1$.
Hence the following contradictory result concludes the proof of this theorem:
\begin{align*}
LA(\overline{T}^{\prime},\overline{\varphi}^{\prime}, \overline{G}) < LA(T,\varphi, G) + 1 = LA(\overline{T}, \overline{\varphi}, \overline{G})
\end{align*}
\end{proof}

\section{Tree Length of Routing Trees}
\label{sect:routing-trees}
  In the previous sections we focused on the problem of embedding vertices of an input graph $G$
into leaf nodes of a host tree $T$, where the degree of every internal node of $T$ is $3$, known as
tree layout problem.
In this section we extend some results to the general routing tree problems.
In this problem the vertices of the source graph $G$ are being embedded into the
leaf nodes of some communication tree $T$ with fixed maximum degree $\Delta$.
\begin{definition}{Minimum Routing Tree Length}
Given graph $G$ and integer $\Delta$, Minimum Routing Tree Length problem (Min Routing Length for short)
is the problem of finding tree $T$ with maximum degree $\Delta$ and a bijective mapping $\varphi:V(G)\rightarrow V_L(T)$,
such that $LA(T,\varphi,G)$ is minimized.
\end{definition}
\noindent
Proof of our final result on Min Routing Length problem is built on some intermediate result as presented in what follows.
\begin{definition}{Fixed Size $k$-Clique Cover} Consider graph $G=(V,E)$ and $k$ positive integers $n_1, \ldots, n_k$
where $\sum_{1 \le i \le k}{n_i} = n$. Fixed Size $k$-Clique Cover problem is the problem of partitioning $V$ into
$k$ disjoint subsets  $V_1, \ldots, V_k$ s.t. $G(V_i)$ is a clique of size $n_i$, for $1 \le i \le k$.
\end{definition}
\begin{lemma}
\label{lemma:fixed-k-clique-cover}
Fixed Size $k$-Clique Cover is NP-complete.
\end{lemma}
\begin{proof}
Similar to the proof of NP-completeness of Fixed Size $k$-Clique Cover problem as a variation of graph $k$-colorability problem.
\end{proof}
\begin{definition}{Equal Size $k$-Clique Cover} Given graph $G=(V,E)$ where $|V| = k\times l$ for some $l \in \mathbb{N}$, Equal Size $k$-Clique Cover problem is the problem of partitioning $V$ into
$k$ disjoint subsets  $V_1, \ldots, V_k$ s.t. $G(V_i)$ is a clique of size $\dfrac{n}{k}$, for $1 \le i \le k$ .
\end{definition}
\begin{lemma}
\label{lemma:equal-k-clique-cover}
Equal Size $k$-Clique Cover problem is NP-complete.
\end{lemma}
\begin{proof}
Similar to the proof of NP-completeness of Equal Size $k$-Clique Cover problem.
\end{proof}
\noindent
Consider the class of graphs where for every graph $G$ in this class, $\exists l \in \mathbb{N}$ such that $|V(G)| = k\times (k-1)^{l}$.
Equal Size $k$-Clique Cover problem stays NP-complete for this class. Concisely in what follows,
a polynomial reduction from Equal Size $k$-Clique Cover problem for general graphs is presented.
Given graph $G=(V,E)$ where $|V| = k\times n^{\prime}$ for $n^{\prime} \in \mathbb{N}$, one can obtain graph $G^{\prime}$ by augmenting
$G$ with $k$ complete components $C_1,\ldots,C_k$ of size $(k-1)^l - n^{\prime}$, where $l$ is the smallest integer such that $(k-1)^{l} \ge n^{\prime}$. Also in $G^{\prime}$ every newly introduces vertex $v \in V(C_1) \uplus \ldots \uplus V(C_k)$ is connected to
every vertex $u \in V(G)$. It is not hard to check that $|V(G^{\prime})| = k\times (k-1)^{l}$ and more importantly
$G$ is a positive instance of Equal Size $k$-Clique Cover problem iff $G^{\prime}$ is a positive instance of Equal Size $k$-Clique Cover problem.

\noindent
In the rest of this section we only consider graph of order $|V(G)| = k\times (k-1)^{l}$. Obviously all the harness results
for this class immediately extend to the class of general sized graphs.
\begin{theorem}
\label{thr:multi-graph-routing-length-np-hard}
Given multi-graph $G$ and integer $\Delta$, the problem of finding a routing tree $T$
and bijective mapping $\varphi: V(G) \rightarrow V_L(T)$ with minimum tree length is NP-hard.
\end{theorem}
\begin{proof}
Similar to the proof of theorem~\ref{thr:multi-graph-tree-length-np-hard}, it can be shown that the
Equal Size $k$-Clique Cover problem is not harder than Minimum Routing Tree Length problem.

\noindent
Hence, consider graph $G$ as the input of the Equal Size $k$-Clique Cover problem where 
$|V(G)| = k\times (k-1)^{l}$ for some $l \in \mathbb{N}$.
Let $G^{\prime}$ be the multi-graph obtained from $G$ by introducing $M = m\times(2n-2)$ parallel edges
between every two vertices $u, v \in V(G)$.
Therefore $G^{\prime} = G \uplus \widetilde{G}$, where complete multi-graph $\widetilde{G}$,
is a vertex induced subgraph of $G^{\prime}$.
Using similar reasoning as in the proof of theorem~\ref{thr:multi-graph-tree-length-np-hard},
$\widetilde{G}$ dictates the structure of optimal routing tree for $G^{\prime}$.
In other words, the problem of finding optimal routing tree $T$ and mapping $\varphi$ for $G^{\prime}$
reduces to the problem of finding mapping $\varphi$ from vertices of original graph $G$
to the leaf nodes of a fixed-structure tree $T$ such that $LA(T, \varphi, G)$ is minimized.
Based on corollary~\ref{coral:opt-tree-layout-complete-graph} $T \in \mathfrak{T}(k, k)$.

\noindent
Removing the only central node of $T \in \mathfrak{T}(k, k)$ partitions $T$
into $k$ subtrees $\mathcal{T}_1, \mathcal{T}_2, \ldots, \mathcal{T}_k \in \mathfrak{T}(k, k)$ of the same order and $(k-1)^l$
leaf nodes.
As explained with more details in the proof of theorem~\ref{thr:multi-graph-tree-length-np-hard},
it can be inferred that $G$ is a positive instance of the Equal Size $k$-Clique Cover problem (in other words
vertices of $G$ can be partitioned into $k$ equal sized complete sub-graphs $G_1,G_2,\ldots,G_k$)
iff given routing tree $T$ and mapping $\varphi$ as the solution for Min Routing Length problem (with $\Delta = k$),
$\varphi$ maps vertices of $G_i$ to leaf nodes of $\mathcal{T}_i$ for $1 \le i \le k$.

\end{proof}

\Hide
{\footnotesize
\printbibliography
}

{\footnotesize 
\bibliographystyle{plainurl} 
\bibliography{allbibs}
}

\ifTR
\else
\fi

\end{document}